\documentclass[11pt]{article}
\usepackage[T1]{fontenc}
\usepackage[utf8]{inputenc}
\usepackage{lmodern}
\usepackage[margin=1in]{geometry}
\usepackage{amsmath,amssymb,amsthm,mathtools,booktabs,graphicx}
\usepackage[authoryear,round]{natbib}
\usepackage{xcolor,comment}
\usepackage{setspace,microtype,flafter,placeins,hyperref}
\hypersetup{hidelinks,
 pdftitle={Drift Inference for Unit-Root Galton--Watson Processes with Immigration},pdfauthor={Yang Lu}}
\newtheorem{assumption}{Assumption}
\newtheorem{theorem}{Theorem}
\newtheorem{proposition}{Proposition}
\newtheorem{lemma}{Lemma}
\newtheorem{corollary}{Corollary}
\theoremstyle{remark}
\newtheorem{remark}{Remark}
\newcommand{\E}{\mathbb E}
\newcommand{\PP}{\mathbb P}
\newcommand{\1}{\mathbf 1}
\newcommand{\F}{\mathcal F}
\newcommand{\Op}{O_p}
\newcommand{\op}{o_p}
\newcommand{\R}{\mathbb R}
\DeclareMathOperator{\Var}{Var}
\title{Drift Inference for Unit-Root Galton--Watson Processes with Immigration}
\author{Yang Lu\\
\small Department of Mathematics and Statistics, Concordia University\\
\small Montreal, Canada\\
\small \texttt{yang.lu@concordia.ca}}
\date{}
\begin{document}
\maketitle
\begin{abstract}
Motivated by weather-related disaster counts in a changing climate,
we study a critical Galton--Watson process with immigration. In such a
model, the drift measures the expected change in the
count from one period to the next. Ordinary least squares is
inconsistent for this drift, so we study state-weighted least squares,
which gives greater weight to observations following small counts,
where conditional variance is lower. Under strict recurrence, the
drift estimator converges at a polynomial rate and the
studentized statistic is asymptotically standard normal. At the recurrence boundary, convergence is logarithmic and the studentized limit is parameter-free
and non-Gaussian. Under recurrence, estimating the optimal offset preserves the
first-order limit obtained with the offset known.
Simulations show lower root mean squared error than
time-weighted least squares with weights $1/t$. Estimating the optimal
weights generally improves point estimation, while imposing the unit
root improves state-weighted confidence-interval coverage with little
change in interval length.
\end{abstract}

\noindent\textbf{Acknowledgements.} We thank NSERC
(RGPIN-2021-04144) and SSHRC (430-2024-01325) for financial support.


\section{Introduction}
\label{sec:intro}

The frequency of weather-related disasters may change with the climate,
motivating nonstationary count models for forecasting
\citep{PeiLu2025,AnneLuYuZhou2026}. A unit root allows unexpected
changes in the count to have lasting effects on forecasts, while
positive drift allows the expected count to grow over time. We study
inference on this drift, which measures the expected change in the
count from one period to the next.

\citet{Michel2020} studies the limiting behavior of a nonstationary
INGARCH model. \citet{BarretoSouzaChan2024} develop asymptotic
inference for the autoregressive coefficient of a near-unit-root
INARCH model, assuming that the drift is known. Neither paper
establishes asymptotic inference for an unknown drift. We consider a
Galton--Watson process with immigration, a flexible framework
for count time series that allows general offspring and immigration
distributions \citep{Lu2021Predictive}. In this process, the count $X_t$
is the size of generation $t$:
\[
 X_t=\sum_{i=1}^{X_{t-1}}Z_{i,t}+\epsilon_t,\qquad t\geq1,
\]
where $Z_{i,t}$ is the latent number of offspring of individual $i$
in generation $t-1$. The variables $Z_{i,t}$ are iid nonnegative
integers with mean $m=1$ and variance $\sigma^2>0$. The latent
immigration counts $\epsilon_t$ are iid with mean $\mu>0$ and variance
$b>0$, and are independent of the offspring array $\{Z_{i,t}\}$.
We assume $\PP(\epsilon_1=0)>0$, so that the process can return
to zero, and take $X_0$ to be a fixed finite nonnegative integer.
With offspring mean one, the process is critical, and its conditional
mean and variance are
\[
 \E(X_t\mid X_{t-1})=X_{t-1}+\mu,\qquad
 \Var(X_t\mid X_{t-1})=\sigma^2X_{t-1}+b,
\]
so the conditional mean equation has a unit root and drift $\mu$.

Two examples of this family are Poisson and negative binomial count
autoregressions. Poisson offspring and immigration give
\[
 X_t\mid X_{t-1}=x\sim\operatorname{Poisson}(x+\mu),
\]
the unit-root version of the Poisson INARCH(1) model studied by
\citet{Weiss2010INARCH}; see also \citet{BarretoSouzaChan2024} for the
near-unit-root case. Write $\operatorname{NB}(r,p)$ for the negative
binomial distribution with shape $r$ and success probability $p$.
Geometric offspring with mean one and negative binomial immigration
with mean $\mu$, both with success probability $1/2$, give
\[
 X_t\mid X_{t-1}=x\sim\operatorname{NB}(x+\mu,1/2).
\]
This is the unit-root version of the negative binomial
autoregressive (NBAR) model \citep{Hopfner1990,GourierouxLu2019}. Here
$\sigma^2=2$ and $b=2\mu$, so the conditional variance is twice the
conditional mean, allowing for overdispersion.

For the general model, the $h$-step conditional mean forecast is
$X_t+h\mu$. An error in estimating $\mu$ is therefore multiplied by $h$
in the forecast mean, making reliable inference on $\mu$ especially
important at longer horizons. In this paper, we study estimation and
inference for $\mu$.

To see why estimating $\mu$ is difficult in this nonstationary setting,
write the model as a regression,
$X_t=mX_{t-1}+\mu+W_t$ with $m=1$, where the innovation $W_t=X_t-X_{t-1}-\mu$ has
conditional mean zero given the observations through $t-1$. The drift
$\mu$ is the intercept and the offspring mean $m$ is the slope,
playing the role of an autoregressive coefficient.
Because the conditional variance $\sigma^2X_{t-1}+b$ grows with the
level, the natural estimators fail. The average increment
$(X_n-X_0)/n$ has a nondegenerate limiting distribution
\citep[Theorem~2.1]{WeiWinnicki1989}, and the {ordinary least squares}
estimator of $\mu$ is inconsistent \citep[p.~1759]{WeiWinnicki1990},
as is its analogue in positively regular unstable INAR(2) models
\citep[Remark~2.10]{BarczyIspanyPap2014}.
\citet{WeiWinnicki1990} therefore proposed weighted least squares (WLS)
with weights $1/(1+X_{t-1})$, which depend on the state $X_{t-1}$; we
call this state-weighted WLS. They proved that the joint estimators of
$m$ and $\mu$ are consistent throughout the critical case.

Inference is harder, because the limiting distribution depends on the
long-run behavior of $X_t$, which is governed by $\tau:=2\mu/\sigma^2$.
Under a mild moment condition, $X_t\to\infty$ almost surely when
$\tau>1$ (transience), whereas for $0<\tau\leq1$ the process returns to
zero infinitely often, with infinite mean return time (null recurrence)
\citep[Corollary~2.11]{WeiWinnicki1989}.
{Within null recurrence, we call $0<\tau<1$ strict recurrence
and $\tau=1$ the recurrence boundary.}
\citet{WeiWinnicki1990}
obtained the limiting distribution of the {drift estimator} only for
$\tau>1$; their Remark~2.6 states that it is unknown for $\tau\leq1$,
and \citet{Wei1991} also treats only the transient case.

To obtain inference in all three regimes, \citet{Lu2026} instead used
the deterministic weights $1/t$. The resulting time-weighted WLS drift
estimator is consistent and asymptotically normal at rate
$\sqrt{\log n}$ for every fixed $\tau>0$.

This common Gaussian form simplifies inference, but
time weighting has two limitations. First, the resulting confidence
intervals shrink only at rate $1/\sqrt{\log n}$, even under recurrence.
Under recurrence, the process repeatedly returns to small counts,
where conditional variance is low. When $X_{t-1}=0$, the offspring sum is empty and
$X_t=\epsilon_t$, directly revealing an immigration count with mean
$\mu$. This observation has the smallest conditional variance, $b$,
and receives the largest state weight, $1/(1+X_{t-1})=1$.
Because $\tau$ governs visits to zero and other low counts, we expect
an estimator that exploits these observations more effectively to
converge faster, at a rate governed by $\tau$.
Second, time weighting discounts observations according to when they
occur rather than their realized count. Two observations with the same
preceding count have the same conditional variance, even if one occurs
much later. State weighting gives them equal weight, whereas time
weighting gives less weight to the later observation.

{

In this paper, we return to the state-weighted approach of
\citet{WeiWinnicki1990}. We extend their weights to
$(a+X_{t-1})^{-1}$ for any fixed offset $a>0$ and obtain the convergence rates
and studentized limits summarized in Table~\ref{tab:regime-guide}.

\begin{table}[!h]
\centering
\caption{State-weighted drift inference across the three regimes.}
\label{tab:regime-guide}
\small
\begin{tabular}{llll}
\toprule
Regime & \shortstack[l]{Convergence\\rate} & \shortstack[l]{Limit of the\\studentized statistic} & Result \\
\midrule
Strict recurrence, $0<\tau<1$ & $n^{(1-\tau)/2}$ & Standard normal & Theorem~\ref{thm:strict}, Proposition~\ref{prop:strict-limit} \\
Recurrence boundary, $\tau=1$ & $\log n$ & Non-Gaussian, parameter-free & Theorem~\ref{thm:boundary} \\
Transience, $\tau>1$ & $\sqrt{\log n}$ & Standard normal & Proposition~\ref{prop:transient-fixed} \\
\bottomrule
\end{tabular}
\end{table}

Under strict recurrence, the random number of returns to zero sets the
scale of the drift estimator's limit. Dividing by an estimated standard
error removes that random scale and gives standard normal inference.
At the recurrence boundary, the limit again has a random scale, but the
estimation error depends on it, so the studentized limit is
non-Gaussian, although free of model parameters.
These results complete the recurrent distributional theory left
open by \citet{WeiWinnicki1990}.

}

We go on to consider two extensions. First, extending the
choice $a=1$ of \citet{WeiWinnicki1990} to arbitrary positive offsets
lets us choose the weights for greater precision. Under strict
recurrence, the optimal offset
 $a_*=b/\sigma^2$ minimizes the asymptotic
scale of the drift estimation error, and, under recurrence, estimating
this offset preserves the first-order limit obtained when it is known.
Second, \citet{WeiWinnicki1990} estimate the slope and the drift
jointly. We show that imposing the unit root $m=1$ instead leaves the
first-order drift limit unchanged.

Simulations under the Poisson and negative binomial models show lower
root mean squared error for state-weighted WLS than for time-weighted
WLS. Estimating the offset improves empirical mean squared error in most
designs, and imposing the unit root improves state-weighted interval
coverage, with little change in interval length.

Section~\ref{sec:main} derives the fixed-offset limits, including those
of the fit that imposes the unit root, and
compares the three regimes. Section~\ref{sec:adaptive} treats estimated
offsets. Section~\ref{sec:simulations} presents a simulation study of
estimator accuracy and confidence-interval performance.
Section~\ref{sec:application} illustrates the methods using global
flood-disaster counts and UK COVID-19 mortality counts.
Proofs are gathered in the appendices.
\section{Fixed-offset estimators and their limits}
\label{sec:main}

Our limit theory needs slightly more than the finite variances
assumed in Section~\ref{sec:intro}.

\begin{assumption}\label{ass:model}
For some $\delta>0$, $\E Z^{2+\delta}+\E\epsilon^{2+\delta}<\infty$,
where $Z$ and $\epsilon$ have the offspring and immigration
distributions.
\end{assumption}

Similar moment conditions are used in this literature; see
\citet[Theorem~2.5]{WeiWinnicki1990} and
\citet[Assumption~5]{AnneLuYuZhou2026}.
The extra $\delta$ serves two purposes here. First, it implies
$\E(Z^2\log^+Z)<\infty$, under which the chain is null recurrent for
$\tau\leq1$ and its probability of being at zero at time $n$ decays
like $n^{-\tau}$ \citep[Lemma~2.10 and Corollary~2.11]{WeiWinnicki1989}.
Second, it supplies the higher-moment bounds used in our martingale
limit theorems.

Given observations $X_0,\ldots,X_n$ and a fixed offset $a>0$, let
$w_a(x)=(a+x)^{-1}$. Define the estimators of the slope $m$ and
drift $\mu$ by
\begin{equation}
 (\widehat m_n(a),\widehat\mu_n(a))
 =\operatorname*{argmin}_{(m,\nu)\in\R^2}
 \sum_{t=1}^n w_a(X_{t-1})(X_t-mX_{t-1}-\nu)^2.
 \label{eq:estimator}
\end{equation}
This extends the estimator of \citet{WeiWinnicki1990} from $a=1$
to arbitrary positive offsets.
If the design is singular, set the slope to one and estimate the drift
by the weighted mean of the increments.
We call eq.~\eqref{eq:estimator} the joint fit, since it
estimates the slope together with the drift. The constrained fit
imposes the unit root $m=1$ and estimates only the drift:
\begin{equation}
 \widehat\mu_{n,\mathrm c}(a)
 =\frac{\sum_{t=1}^n w_a(X_{t-1})(X_t-X_{t-1})}
       {\sum_{t=1}^n w_a(X_{t-1})}.
 \label{eq:constrained-estimator}
\end{equation}
We use these terms for both state and time weighting.

{
Write $V(x)=\sigma^2x+b$ for the conditional variance following a count
$x$. The innovation $W_t=X_t-X_{t-1}-\mu$ has conditional mean zero and variance
$V(X_{t-1})$ given the past. All results below are
driven by three sums:
\begin{equation}
 \begin{aligned}
 H_n(a)&=\sum_{t=1}^n w_a(X_{t-1}),\\
 M_n(a)&=\sum_{t=1}^n w_a(X_{t-1})W_t,\\
 Q_n(a)&=\sum_{t=1}^n w_a(X_{t-1})^2V(X_{t-1}).
 \end{aligned}
 \label{eq:scoreclock}
\end{equation}
These are the accumulated weight, the weighted sum of innovations, and
the sum of the conditional variances of its terms. Since
$\sum_{t=1}^n w_a(X_{t-1})X_{t-1}=n-aH_n(a)$, the intercept normal
equation gives
\begin{equation}
 \widehat\mu_n(a)-\mu
 =\frac{M_n(a)}{H_n(a)}
 -(\widehat m_n(a)-1)\left\{\frac{n}{H_n(a)}-a\right\}.
 \label{eq:reduction}
\end{equation}
The first term is a weighted average of the innovations. The second is
the cost of estimating the slope, which is asymptotically negligible
in all three regimes.
The limit of the estimator is therefore governed by
the joint behavior of $H_n(a)$ and $M_n(a)$; $Q_n(a)$ sets its scale.
We establish the joint behavior of these sums for each recurrent regime
below. For inference, define the fitted residuals
$e_t(a)=X_t-\widehat m_n(a)X_{t-1}-\widehat\mu_n(a)$, suppressing their
dependence on the sample size $n$. Since $Q_n(a)$ involves
the unknown variances $\sigma^2$ and $b$, replace it by its residual
analogue and define the standard error of $\widehat\mu_n(a)$:
\begin{equation}
 \begin{aligned}
 \widehat Q_n(a)&=\sum_{t=1}^n w_a(X_{t-1})^2e_t(a)^2,\\
 \widehat{\mathrm{se}}_n(a)&=\frac{\sqrt{\widehat Q_n(a)}}{H_n(a)}.
 \end{aligned}
 \label{eq:feasibleQ}
\end{equation}
We call $\{\widehat\mu_n(a)-\mu\}/\widehat{\mathrm{se}}_n(a)$ the
studentized statistic. The results below give its limit in each regime.
}

\subsection{Strict recurrence}
\label{sec:strict}

We first state the studentized limit used for inference. The convergence
rate and unstudentized limit are given in Proposition~\ref{prop:strict-limit}
below.

\begin{theorem}[Studentized inference under strict recurrence]\label{thm:strict}
Under Assumption~\ref{ass:model}, if $0<\tau<1$, then for every fixed $a>0$,
\begin{equation}
 \frac{\widehat\mu_n(a)-\mu}{\widehat{\mathrm{se}}_n(a)}
 \Rightarrow N(0,1).
 \label{eq:strict-pivot}
\end{equation}
\end{theorem}

\begin{proof}
See Appendices~\ref{app:strict-cycle} and~\ref{app:residual-studentization}.
\end{proof}

For $0<\eta<1$, let $z_p$ denote the $p$-quantile of the standard normal
distribution. The studentized limit gives a confidence interval with
nominal coverage $1-\eta$:
\begin{equation}
 \widehat\mu_n(a)\ \pm\ z_{1-\eta/2}\,\widehat{\mathrm{se}}_n(a).
 \label{eq:strictCI}
\end{equation}

The rate and the unstudentized limit depend on how often the process
returns to zero and on what each return cycle contributes. Write
$\alpha=1-\tau$, and $\PP_x$ and $\E_x$ for probability and expectation
when $X_0=x$. Let
$p_0=\lim_{n\to\infty}n^\tau\PP_0(X_n=0)\in(0,\infty)$ be the
constant in the return-probability decay noted after
Assumption~\ref{ass:model}; the limit exists for every $0<\tau\leq1$
\citep[Lemma~2.10]{WeiWinnicki1989}.
Let $\pi$ be the invariant measure of $(X_t)$ generated by return cycles
to zero and normalized by $\pi_0=1$, so that $\pi_j$ is the expected
number of visits to state $j$ between successive visits to zero. For
each fixed $a>0$, define
\begin{equation}
 h(a)=\sum_{j\geq0}\pi_jw_a(j),\qquad
 v(a)=\sum_{j\geq0}\pi_jw_a(j)^2V(j),
 \label{eq:hv}
\end{equation}
the expected contributions of one return cycle to $H_n(a)$ and $Q_n(a)$;
although $\pi$ has infinite total mass, both are finite and positive
under Assumption~\ref{ass:model}. Finally, let $\mathcal M_\alpha$ have
the Mittag--Leffler distribution of order $\alpha$, characterized by
$\E e^{-\lambda\mathcal M_\alpha}=\sum_{k\geq0}(-\lambda)^k/\Gamma(1+\alpha k)$
for $\lambda\geq0$; it arises as the limit of the normalized number of
returns to zero.

\begin{proposition}[Rate and unstudentized limit under strict recurrence]\label{prop:strict-limit}
Under Assumption~\ref{ass:model}, if $0<\tau<1$, then for every fixed $a>0$,
\begin{equation}
 n^{(1-\tau)/2}\{\widehat\mu_n(a)-\mu\}
 \Rightarrow
 \sqrt{\frac{v(a)}{p_0\Gamma(\alpha)h(a)^2\mathcal M_\alpha}}\,\mathcal Z,
 \label{eq:strict-fixed}
\end{equation}
where $\alpha=1-\tau$ and $\mathcal Z$ is standard normal and independent of $\mathcal M_\alpha$.
\end{proposition}

\begin{proof}
See Appendix~\ref{app:strict-cycle}.
\end{proof}

Let us give a sketch of the proofs of
Theorem~\ref{thm:strict} and Proposition~\ref{prop:strict-limit}.
Let
\begin{equation}
 N_n=\sum_{t=1}^n\1_{\{X_t=0\}}
 \label{eq:return-count}
\end{equation}
count visits to zero during periods $1,\ldots,n$. Under strict
recurrence, $N_n/\{p_0\Gamma(\alpha)n^\alpha\}\Rightarrow\mathcal M_\alpha$
\citep[Theorem~2]{Pakes1971}. By the ratio limit theorem of
\citet{WeiWinnicki1989}, $H_n(a)$ and $Q_n(a)$ grow like $h(a)N_n$ and
$v(a)N_n$, while $M_n(a)$ fluctuates on the scale $\sqrt{N_n}$. These
occupation limits do not give the limit of the first term in
eq.~\eqref{eq:reduction}, $M_n(a)/H_n(a)$, which requires the joint law
of its numerator and denominator.
Following the regenerative approach of
\citet[Section~4]{HopfnerLocherbach2003}, we obtain the following joint
limit from the discrete-time cycle sums:
\begin{equation}
 \left(\frac{H_n(a)}{n^\alpha},\frac{M_n(a)}{n^{\alpha/2}},
       \frac{Q_n(a)}{n^\alpha}\right)
 \Rightarrow
 \left(p_0\Gamma(\alpha)h(a)\mathcal M_\alpha,
       \sqrt{p_0\Gamma(\alpha)v(a)\mathcal M_\alpha}\,\mathcal Z,
       p_0\Gamma(\alpha)v(a)\mathcal M_\alpha\right).
 \label{eq:strict-joint}
\end{equation}
Dividing the second component by the first gives
eq.~\eqref{eq:strict-fixed}. For the studentized limit, replace $Q_n(a)$
in the third component by $\widehat Q_n(a)$; dividing the second
component by the square root of the third cancels $\mathcal M_\alpha$
and gives eq.~\eqref{eq:strict-pivot}.

\subsection{The recurrence boundary}
\label{sec:boundary}

At the recurrence boundary $\tau=1$, studentization still removes the
unknown model constants, but the limiting distribution is
non-Gaussian.
To define this limit, let $B$ be
standard Brownian motion and set $R(s)=B(s)+\ell(s)$ with
$\ell(s)=-\min_{0\leq u\leq s}B(u)$. Then $R$ is reflected Brownian
motion and $\ell$ is the minimal nondecreasing correction keeping it
nonnegative. Define
\[
 T=\inf\{s:R(s)=1\},\qquad U=\ell(T).
\]

\begin{theorem}[{Drift estimator at the recurrence boundary}]\label{thm:boundary}%
Under Assumption~\ref{ass:model}, if $\tau=1$, then for every fixed $a>0$,
\begin{equation}
 \frac{\widehat\mu_n(a)-\mu}{\widehat{\mathrm{se}}_n(a)}
 \Rightarrow\frac{1-U}{\sqrt T}.
 \label{eq:boundary-pivot}
\end{equation}
Moreover, the drift estimator converges at rate $\log n$, with
\begin{equation}
 \log n\{\widehat\mu_n(a)-\mu\}\Rightarrow\sigma^2\frac{1-U}{T}.
 \label{eq:boundary-est}
\end{equation}
\end{theorem}

\begin{proof}
See Appendices~\ref{app:boundary-potential}, \ref{app:boundary-offsets}, and~\ref{app:residual-studentization}.
\end{proof}

At the boundary, \citet[Remark~2.24]{WeiWinnicki1989} conjectured a
$(\log n)^2$ order for $H_n(1)$, conditional on a harmonic-moment
conjecture of \citet[p.~15]{Pakes1975}, and noted that ``the
convergence type of the possible limiting result is still unknown.''
\citet{LiZhang2019,LiZhang2021} proved Pakes's conjecture; Theorem~1.1
of \citet{LiZhang2021} implies $\E H_n(a)\sim(\log n)^2/\sigma^2$ for
every fixed $a>0$.

Let us give a sketch of the proof of
Theorem~\ref{thm:boundary}.
Building on Li and Zhang's estimates for the probability generating
function of $X_n$, we establish the joint distributional limit
\begin{equation}
 \left(\frac{\sigma^2H_n(a)}{(\log n)^2},\frac{M_n(a)}{\log n}\right)
 \Rightarrow(T,1-U).
 \label{eq:boundary-joint}
\end{equation}
This joint limit identifies the dependence between the
accumulated weight and the weighted sum of innovations. Dividing the
second component by the first gives eq.~\eqref{eq:boundary-est}.
For the studentized limit, note that $Q_n(a)/H_n(a)\to\sigma^2$, because
the accumulated weight is dominated by large counts, where
$V(x)/(a+x)\approx\sigma^2$. Hence $Q_n(a)/(\log n)^2\Rightarrow T$.
Replacing $Q_n(a)$ by $\widehat Q_n(a)$ and dividing the second
component of eq.~\eqref{eq:boundary-joint} by the square root of
$\widehat Q_n(a)/(\log n)^2$ gives eq.~\eqref{eq:boundary-pivot}.

The studentized statistic has the same limiting distribution for all model parameters
and fixed offsets, providing a common calibration for boundary inference.
Let $q_p$ denote
the $p$-quantile of $(1-U)/\sqrt T$. The corresponding interval with
nominal coverage $1-\eta$ is
\[
 \left[\widehat\mu_n(a)-q_{1-\eta/2}\,\widehat{\mathrm{se}}_n(a),
       \widehat\mu_n(a)-q_{\eta/2}\,\widehat{\mathrm{se}}_n(a)\right].
\]
Appendix~\ref{app:boundary-quantiles} describes the
computation of these quantiles.
For a nominal 95\% interval, numerical evaluation gives
$q_{0.025}\approx-1.703$ and $q_{0.975}\approx2.214$, compared with
$\pm1.960$ for the standard normal distribution.
Normal critical values therefore distort one-sided tests:
nominal 5\% tests against larger
and smaller values of $\mu$ have asymptotic sizes of about 9\% and
3\%, respectively.

\subsection{Transience}
\label{sec:transience}

For $a=1$ and $\tau>1$, the Gaussian drift limit was established by
\citet[Theorem~2.5]{WeiWinnicki1990} and
\citet[Theorem~1.2]{Wei1991}. The following proposition extends this
limit to every fixed positive offset and establishes residual
studentization for the joint fit.

\begin{proposition}[Fixed-offset inference under transience]
\label{prop:transient-fixed}
Under Assumption~\ref{ass:model}, if $\tau>1$, then for every fixed $a>0$,
\begin{align}
 \sqrt{\log n}\{\widehat\mu_n(a)-\mu\}
 &\Rightarrow N\!\left(0,\sigma^2\left(\mu-\frac{\sigma^2}{2}\right)\right),
 \label{eq:transient}\\
 \frac{\widehat\mu_n(a)-\mu}{\widehat{\mathrm{se}}_n(a)}
 &\Rightarrow N(0,1).
 \label{eq:fixed-transient-pivot}
\end{align}
\end{proposition}

\begin{proof}
See Appendix~\ref{app:residual-studentization}.
\end{proof}

\subsection{The constrained estimator under a maintained unit root}
\label{sec:maintained}

The preceding results concern the joint fit. For the constrained fit
in eq.~\eqref{eq:constrained-estimator}, whose denominator is $H_n(a)$, the
estimation error is exactly
\begin{equation}
 \widehat\mu_{n,\mathrm c}(a)-\mu=\frac{M_n(a)}{H_n(a)}.
 \label{eq:constrained-error}
\end{equation}
Estimate $Q_n(a)$ from the constrained residuals and define the
corresponding standard error by
\begin{equation}
 \begin{aligned}
 \widehat Q_{n,\mathrm c}(a)&=\sum_{t=1}^n w_a(X_{t-1})^2
       \{X_t-X_{t-1}-\widehat\mu_{n,\mathrm c}(a)\}^2,\\
 \widehat{\mathrm{se}}_{n,\mathrm c}(a)&=\frac{\sqrt{\widehat Q_{n,\mathrm c}(a)}}{H_n(a)}.
 \end{aligned}
 \label{eq:constrained-bracket}
\end{equation}

\begin{corollary}[Maintained-unit-root inference]
\label{cor:maintained}
Under Assumption~\ref{ass:model}, for every fixed $a>0$ and
$\tau>0$, the conclusions of Theorem~\ref{thm:strict},
Proposition~\ref{prop:strict-limit}, Theorem~\ref{thm:boundary}, and
Proposition~\ref{prop:transient-fixed} hold with $\widehat\mu_n(a)$
replaced by $\widehat\mu_{n,\mathrm c}(a)$ and, in the studentized limits,
$\widehat{\mathrm{se}}_n(a)$ by $\widehat{\mathrm{se}}_{n,\mathrm c}(a)$. Moreover,
$\widehat Q_{n,\mathrm c}(a)/Q_n(a)\to_p1$.
\end{corollary}

\begin{proof}
See Appendix~\ref{app:residual-studentization}.
\end{proof}

Thus, fixing $m=1$ leaves the first-order limits unchanged; nevertheless,
removing the slope-estimation term in eq.~\eqref{eq:reduction} may improve
finite-sample inference.

\subsection{Comparison and inference across the three regimes}
\label{sec:weight-comparison}

Corollary~\ref{cor:all-regimes} collects the studentized
limits for every fixed $a>0$. The same statistic applies in every
regime; only its critical values change, at $\tau=1$.

\begin{corollary}[Studentized statistic across the three regimes]
\label{cor:all-regimes}

Under Assumption~\ref{ass:model}, for every fixed $a>0$,
\begin{equation}
 \frac{\widehat\mu_n(a)-\mu}{\widehat{\mathrm{se}}_n(a)}
 \Rightarrow
 \begin{cases}
  N(0,1), & \tau\neq1,\\[2mm]
  (1-U)/\sqrt T, & \tau=1.
 \end{cases}
 \label{eq:transient-pivot}
\end{equation}
Here $T$ and $U$ are defined in Section~\ref{sec:boundary}.
The same convergence holds with $(\widehat\mu_n(a),\widehat{\mathrm{se}}_n(a))$
replaced by $(\widehat\mu_{n,\mathrm c}(a),\widehat{\mathrm{se}}_{n,\mathrm c}(a))$.

\end{corollary}

\begin{proof}

This follows from Theorems~\ref{thm:strict} and~\ref{thm:boundary},
Corollary~\ref{cor:maintained}, and Proposition~\ref{prop:transient-fixed}.

\end{proof}

For comparison, write $\widetilde\mu_n$ for the joint time-weighted
WLS {drift estimator}, with weights $1/t$. Under Assumption~\ref{ass:model},
\citet[Theorem~2]{Lu2026} gives
\begin{equation}
 \sqrt{\log n}(\widetilde\mu_n-\mu)
 \Rightarrow N(0,\sigma^2\mu)
 \label{eq:Lu}
\end{equation}
for every fixed $\tau>0$. This limit provides a common Gaussian inference
formula across the three regimes, but state weighting is asymptotically
more precise in all three.
 In both recurrent regimes it
converges faster, at
$n^{(1-\tau)/2}$ in strict recurrence and $\log n$ at the boundary,
against $\sqrt{\log n}$ for time-weighted WLS. In transience the two share
the $\sqrt{\log n}$ rate, but state weighting has the smaller asymptotic
variance, $\sigma^2(\mu-\sigma^2/2)<\sigma^2\mu$.

{
For state weights, the sum of the one-step conditional variances,
$Q_n(a)$, is of the same order as the accumulated weight $H_n(a)$
for every fixed $a>0$, because $V(x)/(a+x)$ lies between
$\min(\sigma^2,b/a)$ and $\max(\sigma^2,b/a)$. The estimation error is
therefore of order $H_n(a)^{-1/2}$.
 The accumulated weight is of order
$n^{1-\tau}$ under strict recurrence and $(\log n)^2$ at the boundary.
Under transience,
\citet[Lemma~2.4]{WeiWinnicki1990} give
$H_n(1)/\log n\to_p(\mu-\sigma^2/2)^{-1}$, which yields the variance in
eq.~\eqref{eq:transient}. Time weights accumulate
$\sum_{t=1}^n1/t\sim\log n$ in every regime. This explains the rates
above, and why the two schemes share the $\sqrt{\log n}$ rate only
under transience.
}

Both schemes can be read as approximations to
inverse-conditional-variance weighting. The state weight $w_a$ decreases
as the conditional variance along the realized path increases. The time
weight $1/t$ instead approximates the inverse expected conditional
variance, since $\E[V(X_{t-1})]\sim\sigma^2\mu t$.
Section~\ref{sec:adaptive} studies the choice of $a$ that makes the
state weights exactly proportional to the inverse conditional variance.

\FloatBarrier
\section{Estimating the optimal offset}
\label{sec:adaptive}

{
By Proposition~\ref{prop:strict-limit}, under strict recurrence the
offset $a$ enters the limit only through
 the scale factor $v(a)/h(a)^2$, and a smaller factor gives a more
concentrated limit. By the Cauchy--Schwarz inequality,
\[
 \frac{v(a)}{h(a)^2}\geq\Bigl(\sum_{j\geq0}\frac{\pi_j}{V(j)}\Bigr)^{-1},
\]
where the sum is finite under strict recurrence. Equality holds
precisely at $a=a_*:=b/\sigma^2$, which gives inverse-conditional-variance weights,
as in the optimal estimating functions of
\citet[Theorem~3.1]{BaiDurairajan1996}, an approach suggested by
\citet[Section~4.5]{WeiWinnicki1990}. At the boundary,
the offset does not matter to first order. The limit in
Theorem~\ref{thm:boundary}
 is the same for every fixed $a$, because the
accumulated weight is dominated by large counts, where
$(a+x)^{-1}\approx x^{-1}$.

The optimal offset depends on the unknown variances
$\sigma^2$ and $b$, so it must be estimated. By Theorem~\ref{ad:limit}, any positive offset
estimator converging in probability to a fixed positive value gives
the same first-order limit as that value.
 We then construct such
an estimator.
}

\begin{theorem}[Estimated offsets]
\label{ad:limit}
Under Assumption~\ref{ass:model}, suppose $0<\tau\leq1$ and let the
positive statistic $\widehat a_n\to_p a_0$ for a deterministic
$a_0\in(0,\infty)$.
Set $r_n=n^{(1-\tau)/2}$ if $0<\tau<1$ and $r_n=\log n$ if $\tau=1$.
Then
\begin{equation}
 r_n\{\widehat\mu_n(\widehat a_n)-\widehat\mu_n(a_0)\}\to_p0.
 \label{ad:oracleequiv}
\end{equation}

Consequently, the unstudentized drift limits in
Proposition~\ref{prop:strict-limit} and Theorem~\ref{thm:boundary} hold with
$\widehat\mu_n(\widehat a_n)$ in place of $\widehat\mu_n(a_0)$.

With the standard error $\widehat{\mathrm{se}}_n$ defined in
eq.~\eqref{eq:feasibleQ}, the studentized statistic satisfies
\begin{equation}
 \frac{\widehat\mu_n(\widehat a_n)-\mu}
      {\widehat{\mathrm{se}}_n(\widehat a_n)}
 \Rightarrow
 \begin{cases}
  N(0,1),&0<\tau<1,\\[2pt]
  (1-U)/\sqrt T,&\tau=1.
 \end{cases}
 \label{ad:feasiblepivot}
\end{equation}
Furthermore, the same conclusions hold for the constrained estimator
$\widehat\mu_{n,\mathrm c}$ and its standard error
$\widehat{\mathrm{se}}_{n,\mathrm c}$ in Corollary~\ref{cor:maintained}.
\end{theorem}

\begin{proof}
See Appendix~\ref{ad:appendix}.
\end{proof}

{We now construct a particular offset estimator in three
steps: (i) an initial joint fit with $a=1$; (ii) a variance regression
that estimates $\sigma^2$ and $b$ from its residuals; and (iii) a refit
using the ratio of the estimates of $b$ and $\sigma^2$ as the offset.
For step
(ii), we strengthen the moment condition in
Assumption~\ref{ass:model} to finite fourth moments.

\begin{assumption}[Finite fourth moments]\label{ad:fourthmoments}
$\E Z^4+\E\epsilon^4<\infty$.
\end{assumption}
}

{
In step (ii), we use the residual variance regression of
\citet[equations~(1.5)--(1.6)]{Winnicki1991}.\footnote{Winnicki writes
$b^2$ for our $b$.} Let $e_t:=e_t(1)$ denote the residuals from the
initial joint fit at $a=1$. Since
$\E(W_t^2\mid X_{t-1})=\sigma^2X_{t-1}+b$, we regress $e_t^2$ on
$X_{t-1}$ and an intercept. Under the fourth-moment
assumption, the conditional variance of $W_t^2$ is of order
$1+X_{t-1}^2$, which motivates weights $(1+X_{t-1})^{-2}$.
The slope and intercept of this regression
estimate $\sigma^2$ and $b$; we denote them by
$(\widehat{\sigma^2}_n,\widehat b_n)$, with the closed form in
Appendix~\ref{ad:appendix}. The offset estimator used in step (iii) is
\begin{equation}
 \widehat a_n=
 \begin{cases}
  \widehat b_n/\widehat{\sigma^2}_n,
    &\widehat b_n>0,\ \widehat{\sigma^2}_n>0,\\
  1,&\text{otherwise},
 \end{cases}
 \label{ad:offsetpilot}
\end{equation}
where $(\widehat{\sigma^2}_n,\widehat b_n)=(0,0)$ if the regression design
is singular.
Refit with $\widehat a_n$ held fixed, using either the joint estimator
in eq.~\eqref{eq:estimator} or the constrained estimator
in eq.~\eqref{eq:constrained-estimator}. The same initial joint fit and variance
regression are used for both refits.
}

{
Under Assumption~\ref{ad:fourthmoments}, when
$0<\tau\leq1$, \citet[Theorems~3.10 and 3.12]{Winnicki1991} give
\begin{equation}
 (\widehat{\sigma^2}_n,\widehat b_n)\to_p(\sigma^2,b),
 \label{ad:pilotconsistency}
\end{equation}
so $\widehat a_n\to_p a_*$. Applying Theorem~\ref{ad:limit} with
$a_0=a_*$ shows that the estimated-offset fit has the same first-order
drift limit as the fit using the known optimal offset, and that the
studentized limit in eq.~\eqref{ad:feasiblepivot} remains valid. These conclusions
hold for both the joint and constrained estimators.
}

When $\tau>1$, the variance regression does not estimate $b$
consistently \citep[Remark~3.13]{Winnicki1991}.
Appendix~\ref{ad:transient-proofs} establishes that $\widehat b_n$
converges to a random limit and that $\widehat a_n$ converges to a
positive, finite random offset. At the large counts reached under
transience, the weights are insensitive to changes in the offset.
Consequently, the offset need not converge to $a_*$ for the drift
estimator to retain its first-order limit. The following proposition
shows that residual studentization yields asymptotically normal
inference.

\begin{proposition}[Estimated offset in the transient regime]
\label{ad:transientoffset}
Under Assumption~\ref{ad:fourthmoments}, suppose
$\tau>1$, and let $\widehat a_n$ be the offset estimator
in eq.~\eqref{ad:offsetpilot}. Then
\begin{equation}
 \frac{\widehat\mu_n(\widehat a_n)-\mu}
      {\widehat{\mathrm{se}}_n(\widehat a_n)}
 \Rightarrow N(0,1).
 \label{ad:transientstudentized}
\end{equation}
The same convergence holds with $(\widehat\mu_n,\widehat{\mathrm{se}}_n)$
replaced by $(\widehat\mu_{n,\mathrm c},\widehat{\mathrm{se}}_{n,\mathrm c})$, using the same
offset estimate $\widehat a_n$.
\end{proposition}

\begin{proof}
See Appendix~\ref{ad:appendix}.
\end{proof}

Combining Theorem~\ref{ad:limit}, the consistency
in eq.~\eqref{ad:pilotconsistency}, and Proposition~\ref{ad:transientoffset}
shows that, under Assumption~\ref{ad:fourthmoments}, the studentized
limits of Corollary~\ref{cor:all-regimes} continue to hold with $a$
replaced by the estimated offset $\widehat a_n$ of
eq.~\eqref{ad:offsetpilot}, for both the joint and the constrained fit.

\begin{remark}[Subcritical processes]
\label{rem:subcritical}
The same three steps apply when $m<1$. Under finite fourth moments,
\citet[Theorem~3.4]{Winnicki1991} gives $\widehat a_n\to_p a_*$, and
\citet[Theorem~2.2]{AknoucheFrancq2023} show that refitting with
estimated conditional-variance weights attains the first-order
distribution obtained with the optimal weights known, at rate
$\sqrt n$. Under stationarity, the {standard error of $\widehat\mu_n$} must use
the full joint covariance matrix, because slope estimation affects its
first-order uncertainty.
\end{remark}

The three-step fit can be computed without first classifying
the regime, and Remark~\ref{rem:subcritical} shows that it also applies
for $m<1$. Inference, however, depends on the regime. Indeed, the
critical values are standard normal
 except at $\tau=1$
(Section~\ref{sec:boundary}).

\section{Simulation study}
\label{sec:simulations}

\subsection{Simulation design}

We compare the accuracy of state- and time-weighted drift estimators
and the coverage and length of their confidence intervals. We also
examine the effects of estimating the offset and imposing the unit
root, including inference at the recurrence boundary.

We simulate the Poisson INARCH(1) and NBAR models introduced in
Section~\ref{sec:intro}, starting from $X_0=0$. For the Poisson model,
$(\sigma^2,b)=(1,\mu)$ and $\mu=\tau/2$. For the NBAR model,
$(\sigma^2,b)=(2,2\mu)$ and $\mu=\tau$, so the conditional variance
is twice the conditional mean. The optimal offset is $a_*=\mu$ in
both models.

Under strict recurrence, we use $\tau\in\{0.25,0.5,0.8\}$ and sample
sizes $n=100$ and $1{,}000$. At the recurrence boundary,
$\tau=1$, we use the same two sample sizes. Just above the boundary, we
use $\tau=1+1/\log(1000)\approx1.145$, so that $(\tau-1)\log n=1$ at
$n=1{,}000$; this design probes finite-sample coverage there. The offset and maintained-unit-root
comparisons are reported at $n=1{,}000$. Each design uses $5{,}000$
independent paths, with competing estimators evaluated on the same
paths.

Unless stated otherwise, we estimate the slope and drift jointly and
use $a=1$ for state-weighted WLS. Its confidence interval is given in
eq.~\eqref{eq:strictCI}. For the time-weighted WLS estimator
$\widetilde\mu_n$, we use the plug-in interval
\begin{equation}
 \widetilde\mu_n\ \pm\ z_{0.975}
 \sqrt{\frac{\widehat{\sigma^2}_n\widetilde\mu_n}{\log n}},
 \label{sim:eq:luinterval}
\end{equation}
based on the limit in eq.~\eqref{eq:Lu}, where $\widehat{\sigma^2}_n$
is the variance-regression estimate of Section~\ref{sec:adaptive}.
Joint and constrained fits are defined by
eqs.~\eqref{eq:estimator}--\eqref{eq:constrained-estimator} for both
weightings, with $w_a(X_{t-1})$ replaced by $1/t$ for time weighting;
the constrained time-weighted interval replaces $\widetilde\mu_n$ in
eq.~\eqref{sim:eq:luinterval} by the constrained estimate, with the
same $\widehat{\sigma^2}_n$.
Replications with a zero or undefined estimated standard error count
as noncoverage. Mean interval lengths are calculated over the remaining
replications. Monte Carlo uncertainty is calculated across independent
paths, accounting for the pairing of estimators evaluated on the same
path. Reported Monte Carlo intervals are pointwise 95\% intervals.

\subsection{Point estimation and interval performance}
\label{sim:oracle}

State-weighted WLS has lower root mean squared error (RMSE) for drift
estimation than time-weighted WLS in every design shown in
Figure~\ref{sim:fig:rmse}. The relative improvement is greater at
smaller $\tau$ and at $n=1{,}000$ than at $n=100$. All 95\% Monte
Carlo intervals for the RMSE ratio lie below one.

\begin{figure}[!t]
 \centering
 \includegraphics[width=\textwidth]{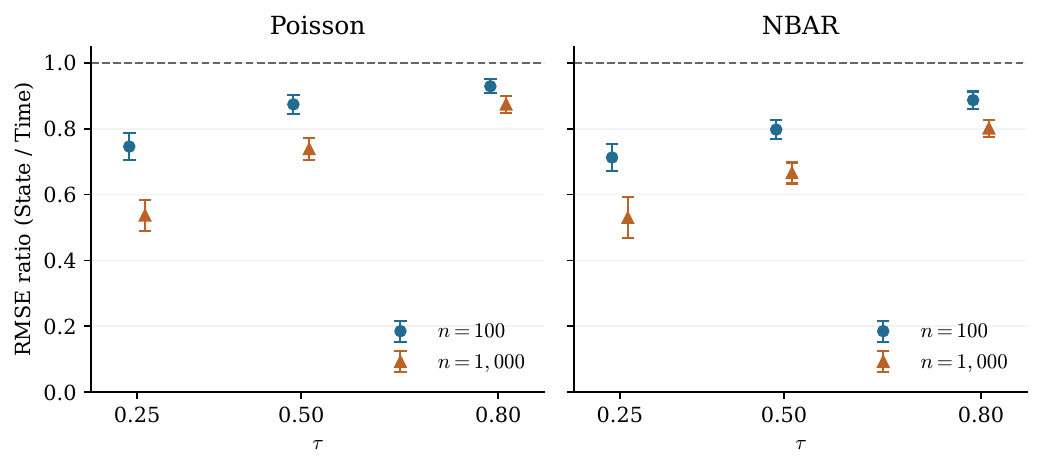}
 \caption{Relative RMSE of drift estimation under strict recurrence in the Poisson and NBAR models.
 State and Time denote joint state-weighted WLS with $a=1$ and joint time-weighted WLS,
 respectively. Points show the State/Time RMSE ratio at $n=100$ and
 $n=1{,}000$. Values below
 the dashed line favor state-weighted WLS. Vertical bars are pointwise
 95\% Monte Carlo intervals from the paired delta method.}
 \label{sim:fig:rmse}
\end{figure}

Table~\ref{sim:tab:matched} first reports coverage with normal
critical values. Joint state-weighted intervals undercover, increasingly
as $\tau$ approaches one, whereas the plug-in time-weighted intervals are
conservative; Section~\ref{sim:maintained} shows that imposing the unit
root largely removes the undercoverage. Because coverage differs, we
compare interval lengths after calibrating both methods to 95\%. Write
$\widehat\mu$ for either drift estimator and $\widehat{\mathrm{se}}$ for
its estimated standard error. Let $q_p^{\mathrm{cal}}$ denote the empirical
$p$-quantile of $(\widehat\mu-\mu)/\widehat{\mathrm{se}}$ in a
separate calibration sample of $5{,}000$ paths.
On the evaluation sample, we then assess the interval
\[
 [\widehat\mu-q^{\mathrm{cal}}_{0.975}\widehat{\mathrm{se}},\,
  \widehat\mu-q^{\mathrm{cal}}_{0.025}\widehat{\mathrm{se}}].
\]
This oracle calibration uses the true model parameters.
The calibration target is adjusted for the frequency of undefined
intervals in the calibration sample.
Monte Carlo intervals for the length ratios use 999 paired bootstrap
resamples of the calibration and evaluation samples.

\begin{table}[htbp]
\centering
\caption{Coverage with normal critical values and oracle-calibrated relative mean interval length, $n=1{,}000$.}
\label{sim:tab:matched}
\footnotesize
\setlength{\tabcolsep}{3.5pt}
\begin{tabular}{llrrrrr}
\toprule
Model & $\tau$ & \multicolumn{2}{c}{\shortstack{Coverage with normal\\critical values (\%)}} & \multicolumn{2}{c}{Calibrated coverage (\%)} & Mean length ratio (\%) \\ \cmidrule(lr){3-4}\cmidrule(lr){5-6} & & State & Time & State & Time & State/Time [95\% MC interval] \\
\midrule
Poisson & 0.25 & {92.72} & 98.72 & {95.56} & 95.30 & 32.56 [31.16, 34.39] \\
Poisson & 0.50 & {88.46} & 97.12 & {94.56} & 94.94 & 47.12 [45.55, 49.09] \\
Poisson & 0.80 & {85.92} & 95.98 & {94.68} & 93.90 & 66.46 [63.95, 69.20] \\
NBAR & 0.25 & {92.40} & 98.54 & {94.42} & 94.78 & 30.11 [28.89, 31.53] \\
NBAR & 0.50 & {89.86} & 97.54 & {94.72} & 95.04 & 45.84 [44.21, 47.78] \\
NBAR & 0.80 & {87.94} & 96.14 & {94.90} & 94.78 & 64.41 [62.13, 66.92] \\
\bottomrule
\end{tabular}
\par\smallskip
\begin{minipage}{\textwidth}\footnotesize State and Time denote WLS with weights $(1+X_{t-1})^{-1}$ and $1/t$, respectively. All fits are joint. Time-weighted WLS uses plug-in standard errors. Equal-tailed quantiles of the studentized statistic are estimated from an independent calibration sample. Length-ratio intervals use paired bootstrap resamples of both samples.\end{minipage}
\end{table}

After calibration, state weighting gives substantially
shorter intervals at comparable coverage.

\FloatBarrier
\subsection{Estimating the offset}
\label{sim:offset}

We next examine how much of the improvement available from the
optimal offset is retained when the offset must be estimated. We
compare $a=1$, the known optimal value $a_*$, and the estimate
$\widehat a_n$ from Section~\ref{sec:adaptive}.

\begin{table}[htbp]
\centering
\caption{Drift RMSE for fixed, optimal, and estimated offsets, $n=1{,}000$.}
\label{sim:tab:adaptive}
\footnotesize
\setlength{\tabcolsep}{3.5pt}
\begin{tabular}{llrrrr}
\toprule
Model & $\tau$ & \multicolumn{3}{c}{RMSE ($\times100$)} & MSE reduction \\ \cmidrule(lr){3-5} & & $a=1$ & $a=a_*$ & $a=\widehat a_n$ & (\%) [95\% MC interval] \\
\midrule
Poisson & 0.25 & {9.48} & {5.70} & {7.43} & {$38.57\ [31.97,\ 45.18]$} \\
Poisson & 0.50 & {18.56} & {15.14} & {16.64} & {$19.66\ [16.35,\ 22.98]$} \\
Poisson & 0.80 & {27.40} & {24.79} & {26.38} & {$7.29\ [5.29,\ 9.29]$} \\
NBAR & 0.25 & {19.07} & {14.86} & {17.85} & {$12.40\ [4.08,\ 20.72]$} \\
NBAR & 0.50 & {31.90} & {29.19} & {31.10} & {$4.91\ [0.85,\ 8.97]$} \\
NBAR & 0.80 & {49.21} & {48.56} & {48.60} & {$2.46\ [-0.62,\ 5.55]$} \\
\bottomrule
\end{tabular}
\par\smallskip
\begin{minipage}{\textwidth}\footnotesize Here $\mathrm{MSE}(a)$ is the empirical mean squared error of the drift estimator fitted with offset $a$. MSE reduction is $100\{1-\mathrm{MSE}(\widehat a_n)/\mathrm{MSE}(1)\}$, with a paired delta-method 95\% Monte Carlo interval. The known optimal offset is $a_*=b/\sigma^2$. All fits use joint state-weighted WLS on the same evaluation paths as Figure~\ref{sim:fig:rmse}.\end{minipage}
\end{table}

Estimating the offset reduces empirical MSE relative to
$a=1$ in all six designs, with the largest reductions in the Poisson
model. For NBAR at $\tau=0.8$, the improvement is small and its Monte
Carlo interval includes zero.

\FloatBarrier
\subsection{Inference under the maintained unit root}
\label{sim:maintained}

Although joint and constrained estimation have the same first-order
limits, estimating the slope can affect finite-sample coverage. We
compare confidence intervals obtained by estimating $m$ jointly with
$\mu$ and by fixing $m$ at its true value of one. The comparison
includes state-weighted WLS with $a=1$ or an estimated offset, and
time-weighted WLS.

\begin{table}[htbp]
\centering
\caption{State-weighted WLS: normal-interval coverage under the maintained unit root, $n=1{,}000$.}
\label{sim:tab:maintained}
\footnotesize
\setlength{\tabcolsep}{3.5pt}
\begin{tabular}{llrrrr}
\toprule
Model & $\tau$ & \multicolumn{2}{c}{$a=1$} & \multicolumn{2}{c}{$a=\widehat a_n$} \\ \cmidrule(lr){3-4}\cmidrule(lr){5-6} & & Joint & Constrained & Joint & Constrained \\
\midrule
Poisson & 0.25 & {92.72} & {92.84} & {93.16} & {93.36} \\
Poisson & 0.50 & {88.46} & {93.34} & {89.16} & {92.54} \\
Poisson & 0.80 & {85.92} & {94.00} & {85.32} & {92.14} \\
NBAR & 0.25 & {92.40} & {93.10} & {92.08} & {92.86} \\
NBAR & 0.50 & {89.86} & {93.78} & {88.90} & {91.88} \\
NBAR & 0.80 & {87.94} & {94.48} & {85.16} & {91.46} \\
\bottomrule
\end{tabular}
\par\smallskip
\begin{minipage}{\textwidth}\footnotesize Entries are coverage percentages for nominal 95\% intervals. State-weighted WLS intervals with $a=1$ or the estimated offset use fitted residuals. Joint fits estimate the slope and drift; constrained fits impose $m=1$. Estimated offsets use the same initial joint state-weighted WLS fit and variance regression in both columns.\end{minipage}
\end{table}

Table~\ref{sim:tab:maintained} shows that imposing the unit root markedly
improves state-weighted interval coverage near the boundary,
while mean interval lengths change by less than 0.2\%.
The observed coverage gains suggest a higher-order benefit from imposing
the unit root.
With the estimated offset, joint-fit coverage shows no consistent
improvement, and constrained-fit coverage is lower than with $a=1$ in
five of the six designs.
The improvement extends to the transient design just above
the boundary, where imposing the unit root raises coverage from
85--87\% to about 94\%.

For time-weighted WLS, by contrast, imposing the unit root shortens
intervals but reduces coverage, to about 90\% at $\tau=0.25$. Just
above the boundary, joint time-weighted intervals have coverage close
to nominal.

\FloatBarrier
\subsection{Boundary inference}
\label{sim:boundary}

At $\tau=1$, we construct state-weighted intervals with $a=1$ using
the parameter-free boundary quantiles from
Section~\ref{sec:boundary}. Table~\ref{sim:tab:boundary}
compares joint and constrained fits with normal intervals from
constrained time-weighted WLS.

\begin{table}[htbp]
\centering
\caption{Coverage and relative mean interval length at the recurrence boundary, $\tau=1$.}
\label{sim:tab:boundary}
\footnotesize
\setlength{\tabcolsep}{3.5pt}
\begin{tabular}{llrrrr}
\toprule
Model & $n$ & \shortstack{State\\Joint} & \shortstack{State\\Constrained} & \shortstack{Time\\Constrained} & \shortstack{Mean length ratio (\%)\\(95\% MC interval)} \\
\midrule
Poisson & 100 & {85.38} & {93.68} & 92.44 & 72.73 [72.24, 73.21] \\
Poisson & 1,000 & {88.60} & {94.86} & 94.32 & 62.60 [62.14, 63.05] \\
NBAR & 100 & {86.66} & {92.78} & 92.66 & 71.55 [70.88, 72.22] \\
NBAR & 1,000 & {89.74} & {94.58} & 94.40 & 62.21 [61.70, 62.72] \\
\bottomrule
\end{tabular}
\par\smallskip
\begin{minipage}{\textwidth}\footnotesize Coverage entries are percentages for nominal 95\% intervals. State and Time denote state-weighted and time-weighted WLS, respectively. State-weighted WLS uses $a=1$ and the boundary quantiles from Section~\ref{sec:boundary}; time-weighted WLS uses plug-in normal intervals. The length ratio is State/Time for constrained fits, in percent, with a paired delta-method 95\% Monte Carlo interval.\end{minipage}
\end{table}

At the boundary, joint state-weighted intervals continue to
undercover. Imposing the unit root brings coverage closer to nominal,
although some undercoverage remains at $n=100$. Constrained
state-weighted intervals are about a third shorter than constrained
time-weighted intervals.

\FloatBarrier
\section{Empirical illustrations}
\label{sec:application}

We use annual global flood-disaster counts and daily UK COVID-19
mortality counts to examine how weighting and the unit-root restriction
affect drift estimation.
Floods illustrate the motivating application; the UK series
allows a comparison with \citet{BarretoSouzaChan2024}, who test for a
unit root in UK COVID-19 deaths over the same period.
Throughout this section, we omit sample-size subscripts on
estimates and suppress the offset argument when its value is given in
the text or table row.

\paragraph{Flood counts.}

We construct annual counts from the EM-DAT archive for 1900--2024,
giving 125 observations and $n=124$ fitted transitions.\footnote{Source:
\href{https://doi.org/10.14428/DVN/I0LTPH}{The EM-DAT Emergency Events Database Archive},
snapshot released April 30, 2026.
We select records with disaster type Flood and identify events by the
first nine characters of \texttt{DisNo.}. An event affecting several
countries is counted once, in its earliest recorded starting year.
Years without a recorded flood are assigned zero.}
Figure~\ref{fig:applications}A shows small recorded counts early in
the sample and larger, more variable counts in later decades.

\paragraph{COVID-19 death counts.}

\citet{BarretoSouzaChan2024} analyze 492 daily UK COVID-19 death counts
from January 30, 2020 to June 4, 2021 with a near-unit-root Poisson
INARCH model. For their unit-root test, they estimate the drift by
Poisson conditional maximum likelihood, obtaining $0.269$, and then
hold that estimate fixed while fitting the slope by ordinary least
squares. They
report a slope estimate of $0.997$ and a unit-root test $p$-value of
$0.704$. Since their theory treats the drift as known, we
revisit the same calendar window to estimate the drift and quantify its
uncertainty under a maintained unit root,
using the UKHSA archive of
national-statistics deaths by date of death.\footnote{Source:
\href{https://ukhsa-dashboard.data.gov.uk/covid-19-archive-data-download}{UKHSA coronavirus dashboard archive}.
We select the United Kingdom series
\texttt{newDailyNsoDeathsByDeathDate}.}
The variance regression of Section~\ref{sec:adaptive} gives
$\widehat{\sigma^2}=3.005$, well above the value of one implied by the
Poisson model, so a Poisson specification understates the conditional
variance, particularly at large counts. State-weighted inference does
not rely on that specification, because it uses only the conditional
mean and variance and computes its standard error from the residuals.

Table~\ref{tab:applications} reports slope and drift
estimates for both series, with nominal 95\% intervals for $\mu$ that
use normal critical values (valid when $m=1$ and $\tau\neq1$).
The plug-in estimates
$\widehat\tau=2\widehat\mu_{\mathrm c}(1)/\widehat{\sigma^2}$ are $0.30$
for floods and $0.27$ for UK deaths, both below one.
For the estimated-offset rows, both joint and constrained
fits use the offset obtained from the initial joint fit with $a=1$.
State-weighted intervals use the standard error
$\widehat{\mathrm{se}}_n(a)$, as in eq.~\eqref{eq:strictCI}, and
$\widehat{\mathrm{se}}_{n,\mathrm c}(a)$ for constrained fits. Time-weighted intervals use
eq.~\eqref{sim:eq:luinterval}.

\begin{table}[htbp]
\centering\small
\setlength{\tabcolsep}{4pt}
\caption{Joint and constrained estimates with nominal 95\% intervals for $\mu$.}
\label{tab:applications}
\begin{tabular}{lrrrrr}
\toprule
 & \multicolumn{3}{c}{Joint fit} & \multicolumn{2}{c}{Constrained fit} \\
\cmidrule(lr){2-4}\cmidrule(lr){5-6}
Weighting & Slope & Drift & 95\% interval
 & Drift & 95\% interval \\
\midrule
\multicolumn{6}{l}{Global floods, 1900--2024 ($n=124$)} \\
State, $a=1$ & $1.008$ & $0.692$ & $[0.269,\ 1.115]$ & $0.711$ & $[0.288,\ 1.134]$ \\
State, $a=\widehat a_n$ & $1.007$ & $0.737$ & $[0.376,\ 1.097]$ & $0.740$ & $[0.379,\ 1.101]$ \\
Time & $0.993$ & $0.243$ & $[-0.713,\ 1.199]$ & $0.170$ & $[-0.630,\ 0.969]$ \\
\addlinespace
\multicolumn{6}{l}{UK deaths, January 30, 2020--June 4, 2021 ($n=491$)} \\
State, $a=1$ & $0.999$ & $0.411$ & $[0.128,\ 0.693]$ & $0.398$ & $[0.117,\ 0.680]$ \\
State, $a=\widehat a_n$ & $0.999$ & $0.216$ & $[0.054,\ 0.378]$ & $0.215$ & $[0.053,\ 0.377]$ \\
Time & $1.001$ & $1.142$ & $[-0.316,\ 2.601]$ & $1.350$ & $[-0.236,\ 2.935]$ \\
\bottomrule
\end{tabular}
\vspace{4pt}
\begin{minipage}{\textwidth}\footnotesize
Constrained fits impose $m=1$. State weights are $(a+X_{t-1})^{-1}$.
Time weights are $1/t$, with $t=1$ at the first fitted transition.
All interval endpoints are reported without truncation at zero.
\end{minipage}
\end{table}
\FloatBarrier

Comparing the joint and constrained drift estimates shows that imposing $m=1$ changes the
state-weighted drift estimates little in either dataset. Their interval
lengths are also nearly unchanged. The time-weighted drift estimates
change more, decreasing for floods and increasing for UK deaths.

\begin{figure}[tbp]
\centering
\includegraphics[width=\textwidth]{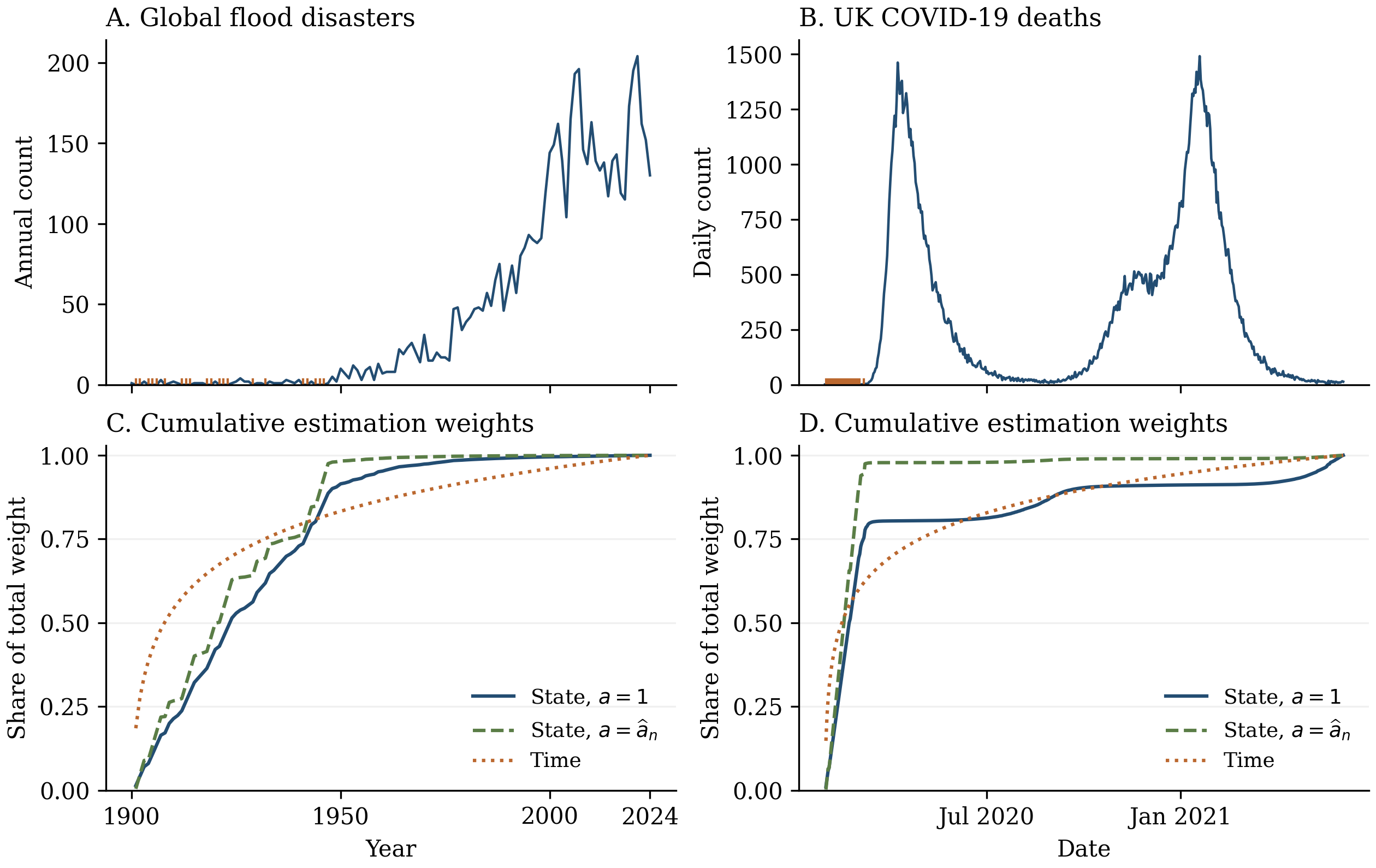}
\caption{Count series and cumulative normalized estimation weights.
Panels A and B show annual global flood-disaster counts and daily UK
COVID-19 death counts; orange ticks mark zero counts. Panels C and D
show state weights with $a=1$ or the estimated offset, and time weights
$1/t$. Each sequence of weights sums to one.}
\label{fig:applications}
\end{figure}

We next compare weighting schemes. By
eq.~\eqref{eq:constrained-estimator}, each constrained estimate is a
weighted average of the increments $X_t-X_{t-1}$, with weights
$(a+X_{t-1})^{-1}$ under state weighting and $1/t$ under time
weighting, so these estimates show the role of the weights most
directly.
 State weights concentrate on
transitions following small counts (Figure~\ref{fig:applications}C,D);
at $a=1$, transitions following zero counts carry 59\% of the state
weight for floods and 71\% for UK deaths (21 and 31 transitions,
respectively). After a zero, the increment is the new count, whose
average is 0.76 for floods and 0.19 for UK deaths. The state-weighted
estimates move toward these averages, which largely explains why state
weighting gives larger drift estimates than time weighting for floods
but smaller ones for UK deaths, in both joint and constrained fits.

The estimated offsets are $0.119$ for floods and $0.077$ for
UK deaths, both below the baseline $a=1$. They raise the share of
transitions following zero counts to 89\% and 96\% of the state weight,
so the constrained estimates move almost onto the average counts
following a zero, which nearly halves the UK estimate but changes the
flood estimate only slightly.
 The estimated-offset intervals are also the
shortest of the three weighting schemes for both joint and constrained
fits in both datasets.

Transitions following zero counts occur mostly early in both
samples (Figure~\ref{fig:applications}). For floods, these early
records are incomplete, and EM-DAT recommends using data from 2000
onward for trend analysis.\footnote{See the EM-DAT documentation on
\href{https://doc.emdat.be/docs/known-issues-and-limitations/specific-biases/}{temporal reporting biases}.}
Table~\ref{tab:applications-start} therefore examines how the
estimates change with the starting date.

\begin{table}[htbp]
\centering\small
\caption{Sensitivity to the starting period: joint-fit slopes and
constrained drift estimates.}
\label{tab:applications-start}
\begin{tabular}{lrrrr}
\toprule
 & & \multicolumn{3}{c}{Constrained drift estimate} \\
\cmidrule(lr){3-5}
Sample period & $\widehat m(1)$ & State, $a=1$ & State, $a=\widehat a_n$ & Time \\
\midrule
\multicolumn{5}{l}{Global floods} \\
1900--2024 & $1.008$ & $0.711$ & $0.740$ & $0.170$ \\
1950--2024 & $0.972$ & $3.046$ & $2.798$ & $0.195$ \\
1980--2024 & $0.906$ & $4.233$ & $4.233^{\dagger}$ & $2.942$ \\
2000--2024 & $0.390$ & $2.110$ & $2.110^{\dagger}$ & $1.413$ \\
\addlinespace
\multicolumn{5}{l}{UK deaths (endpoint June 4, 2021)} \\
From January 30, 2020 & $0.999$ & $0.398$ & $0.215$ & $1.350$ \\
From March 1, 2020 & $0.996$ & $1.175$ & $1.175^{\dagger}$ & $4.718$ \\
\bottomrule
\end{tabular}
\vspace{4pt}
\begin{minipage}{\textwidth}\footnotesize
The slope is from the joint state-weighted fit with $a=1$.
Drift estimates impose $m=1$ in every window; {time weights restart}
at $t=1$. $\dagger$ denotes the fallback $a=1$ in
eq.~\eqref{ad:offsetpilot}, because at least one variance coefficient is
nonpositive.
\end{minipage}
\end{table}

Table~\ref{tab:applications-start} shows that drift estimates in both
applications are sensitive to the starting period, but the slopes
behave differently. For floods, the slope falls to $0.390$ over
2000--2024.
For UK deaths, the slope remains
near one after the earliest dates are removed, even as all three
constrained drift estimates rise. Thus the flood results raise a
question about the unit-root approximation in recent years, whereas
the UK results show that a near-unit-root slope can coexist with a
start-date-sensitive drift estimate.

Finally, we examine autocorrelation in the constrained-fit residuals
with $a=1$ after division by $\sqrt{1+X_{t-1}}$. The flood residuals
have negative lag-one autocorrelation ($-0.204$), whereas the UK death
residuals have positive autocorrelations at lags two to five
($0.173$--$0.268$).

\FloatBarrier
\section{Concluding remarks}
\label{sec:conclusion}

In this paper, we have proposed a solution to the problem that drift
inference for unit-root count series was either unavailable or slow.
Ordinary least squares is inconsistent for the drift; the state-weighted
estimator of \citet{WeiWinnicki1990} is consistent, but its limiting
distribution under recurrence was unknown; and the time-weighted
estimator of \citet{Lu2026} converges only at rate $\sqrt{\log n}$. We
have shown that, when the process is recurrent, the state-weighted
drift estimator converges much faster, at a polynomial rate below the
recurrence boundary and at rate $\log n$ on it.
 The reason is the zero
state. Each return to zero reveals a fresh immigration count with the
smallest conditional variance, and state weighting gives these
observations the largest weight. Below the boundary, precision is
therefore governed by the number of returns to zero rather than by the
length of the series. In the empirical illustrations,
too, most of the state weight falls on the few transitions out of
zero, so the credibility of the drift estimate rests on those
observations.

Away from the boundary, dividing by the estimated standard error
removes this random precision, so standard normal critical values
apply.
 At the boundary, the studentized
limit is non-Gaussian but free of model parameters, so a single set of
critical values serves all models.
Estimating the optimal offset and imposing the unit root leave
these limits unchanged; in simulations, the former improves point
estimation and the latter interval coverage.

Several questions are open. Because the studentized statistic has
different limits at and away from the boundary, intervals with uniform
coverage near it are still needed. The theory also assumes an offspring
mean of exactly one; letting it approach one would connect the
stationary and unit-root cases, building on \citet{AnneLuYuZhou2026} for
affine processes and \citet{BarretoSouzaChan2024} for near-unit-root
INARCH processes. The effect of a preliminary unit-root test on drift
inference is not yet known. Finally, the residual autocorrelation in
both applications points to models with more lags; \citet{LuIspany2026}
treat unit-root INAR(2) models with time-weighted WLS, and state
weighting could be extended to them.

\paragraph{Data availability.}
The global flood-disaster records are available from the EM-DAT archive,
and the UK COVID-19 death counts from the UKHSA coronavirus dashboard
archive. Section~\ref{sec:application} gives the source links, sample
periods, and construction rules.

\paragraph{Use of generative AI.}
OpenAI Codex was used to assist with mathematical
exploration, manuscript revision, and simulation programming.

\setlength{\bibsep}{2pt}
\bibliographystyle{abbrvnat}
\begingroup
\interlinepenalty=10000
\bibliography{references}
\endgroup

\clearpage
\appendix
\setcounter{equation}{0}
\renewcommand{\theequation}{a\arabic{equation}}
\renewcommand{\theHequation}{appendix.\arabic{equation}}
\section{Proofs for fixed-offset estimators}
\label{app:original}

Let $\F_t=\sigma(X_0,\ldots,X_t)$ be the observation filtration.
The innovations $(W_t)$ form a martingale difference sequence with
respect to $(\F_t)$. The weighted sum $M_n(a)$ is a martingale with
predictable quadratic variation $\langle M(a)\rangle_n=Q_n(a)$.
In the proofs, we call $M_n(a)$ the weighted innovation score,
$H_n(a)$ the accumulated-weight clock, and $Q_n(a)$ the score's bracket,
or predictable quadratic variation. A cycle reward is the sum of the
specified contributions over one return cycle to zero.

We first collect the common bounds and control the effect of estimating
the slope. We then prove the joint score and clock limits under strict
recurrence and at the boundary, extend the boundary limit to every fixed
offset, and justify residual studentization. The final subsection
derives the boundary transform and describes the computation of its
quantiles.

For the slope reduction and the strict-recurrence proof, fix $a>0$ and
abbreviate $H_n=H_n(a)$, $M_n=M_n(a)$, and
$Q_n=Q_n(a)$. We likewise suppress the offset
in the fitted coefficients.
The letter $C$ denotes a finite positive constant that may change from
line to line. Constants in these arguments may depend on the fixed offset.

\subsection{Common inputs and the slope reduction}
\label{app:slope-reduction}

\paragraph{Innovation moments.}
Under Assumption~\ref{ass:model}, for every
$2<p\leq\min\{2+\delta,3\}$,
\begin{equation}
 \E\bigl(|W_t|^p\mid\F_{t-1}\bigr)
 \leq C(1+X_{t-1})^{p/2}.
 \label{eq:innovation-moment}
\end{equation}
This bound holds in all three regimes. Conditional on $X_{t-1}=x$,
$W_t=\sum_{i=1}^x(Z_{i,t}-1)+(\epsilon_t-\mu)$.
The moment inequality for sums of independent centered variables bounds
the $p$th moment of the offspring sum by $C(x+x^{p/2})$;
the centered immigration term has a finite $p$th moment.

\paragraph{Functional limit.}
Write $S_n=\sum_{t=1}^nX_{t-1}$. For every $\tau>0$, the critical
functional limit of \citet[Theorem~2.1]{WeiWinnicki1989} implies
\begin{equation}
 X_n-X_0-n\mu=\sum_{t=1}^nW_t=\Op(n),\qquad
 S_n/n^2\Rightarrow J=\int_0^1Y(s)\,ds.
 \label{eq:critical-functional-inputs}
\end{equation}
Here $J>0$ almost surely, and $Y$ solves
$dY(s)=\mu\,ds+\sigma\sqrt{Y(s)}\,dB_Y(s)$ with $Y(0)=0$,
where $B_Y$ is standard Brownian motion.

\paragraph{Recurrent occupation bounds.}
For $0<\tau\leq1$, by
\citet[Lemmas~2.6 and 2.8 and Corollary~2.11]{WeiWinnicki1989},
we may work on the recurrent communicating class containing zero and
start the process at zero; any fixed finite initial population gives
the same limits by their coupling argument. Assumption~\ref{ass:model}
implies $\E[Z^2\log^+Z]<\infty$, so their Lemma~2.10 gives
$u_j:=\PP_0(X_j=0)\sim p_0j^{-\tau}$. We use the return-cycle
invariant measure normalized by $\pi_0=1$.

For the process started at zero, let
$\mathcal S_0=0<\mathcal S_1<\cdots$ be its successive return times.
Then $N_n$ from eq.~\eqref{eq:return-count} counts completed return cycles
by time $n$.

For a nonnegative function $\varphi$ and a set of states $A$, write
$\pi(\varphi)=\sum_{j\geq0}\pi_j\varphi(j)$ and
$\pi(A)=\sum_{j\in A}\pi_j$.

Under Assumption~\ref{ass:model}, the positive-real generating-function
bound of \citet[Proposition~4.1]{LiZhang2019}, restated in
\citet[Lemma~2.2, equation~(2.3)]{LiZhang2021}, gives
$\E_0e^{-sX_t}\leq C(1+ts)^{-\tau}$ for $0<s\leq1$.
For $q\geq1$, integrate this bound using
\[
 \Gamma(q)\E_0(1+X_t)^{-q}
 =\int_0^\infty s^{q-1}e^{-s}\E_0e^{-sX_t}\,ds.
\]
Splitting at $(t+1)^{-1}$ and $1$, and then summing over $t$, gives
\begin{equation}
 \E\sum_{t=1}^n(1+X_{t-1})^{-q}
 \leq C_q
 \begin{cases}
 n^{1-\tau},&0<\tau<1,\quad q\geq1,\\
 (\log n)^2,&\tau=1,\quad q=1,\\
 \log n,&\tau=1,\quad q>1,
 \end{cases}
 \qquad n\geq2.
 \label{eq:harmonic-bound}
\end{equation}
Monotone coupling with a process started at zero extends the upper bound
to every nonnegative initial population, so eq.~\eqref{eq:harmonic-bound}
holds for the model considered here.

\paragraph{Slope reduction.}
Under recurrence, $H_n\to\infty$ almost surely and $Q_n\leq C_aH_n$.
The martingale strong-law argument in the proof of
\citet[Theorem~2.5]{WeiWinnicki1990} therefore applies to $w_a$ and
gives $M_n/H_n\to0$ almost surely. Solving the normal equations gives
\begin{equation}
 \widehat m_n-1
 =\frac{X_n-X_0-n\mu-nM_n/H_n}{S_n+an-n^2/H_n}.
 \label{eq:slope-normal}
\end{equation}
The denominator divided by $n^2$ converges weakly to $J>0$, so
$n(\widehat m_n-1)=\Op(1)$.

\subsection{The cycle limit under strict recurrence}
\label{app:strict-cycle}

Throughout this subsection, $0<\tau<1$ and $\alpha=1-\tau$.
\citet[Section~3.1]{HopfnerLocherbach2003} prove the continuous-time
version of the joint limit in eq.~\eqref{eq:strict-joint} for martingales
of Harris processes with life cycles. We adapt their argument to the
discrete-time chain, with cycles given by returns to zero.

The proof of
\citet[Theorem~2.18]{WeiWinnicki1989} gives $\pi(w_1)<\infty$.
For every fixed $a>0$, $w_a\asymp w_1$ and
$w_a^2V\leq C_aw_1$, hence $0<h(a),v(a)<\infty$.
For the return count $N_n$ in eq.~\eqref{eq:return-count}, their
Corollary~2.20 compares accumulated rewards with visits to zero and gives
\begin{equation}
 \left(\frac{H_n(a)}{N_n},\frac{Q_n(a)}{N_n}\right)
 \longrightarrow(h(a),v(a))\quad\text{almost surely}.
 \label{eq:occupation-ratios}
\end{equation}

\begin{proof}[Proof of Proposition~\ref{prop:strict-limit} and the joint limit]
We verify the cycle moments and normalization needed to apply
\citet[Theorem~3]{ResnickGreenwood1979} to the discrete-time cycle sums.
Work on the recurrent class containing the atom zero, as in
Appendix~\ref{app:slope-reduction}, and define
\[
 \mathcal L_i=\mathcal S_i-\mathcal S_{i-1},\qquad
 \mathcal C_i=\sum_{t=\mathcal S_{i-1}+1}^{\mathcal S_i}w_a(X_{t-1})W_t.
\]
The strong Markov property makes the cycle vectors
$(\mathcal L_i,\mathcal C_i)$ iid. The score is an additive martingale
with additive predictable bracket $Q_n(a)$, whose expected increment
over a cycle is $v(a)<\infty$. Stopping first at
$\mathcal L_i\wedge j_0$ and then taking the $L^2$ limit gives
\begin{equation}
 \E\mathcal C_i=0,\qquad \E\mathcal C_i^2=v(a).
 \label{eq:cycle-isometry}
\end{equation}

The renewal identity and $u_j\sim p_0j^{\alpha-1}$ give
\[
 1-\E e^{-s\mathcal L_1}
 =\left(\sum_{j\geq0}u_je^{-sj}\right)^{-1}
 \sim\frac{s^\alpha}{p_0\Gamma(\alpha)}\qquad(s\downarrow0).
\]
The power-series and Laplace Tauberian theorems
\citep[Sections~1.7 and~8.1]{BinghamGoldieTeugels1987} therefore yield
\[
 \PP(\mathcal L_1>n)
 \sim\frac{n^{-\alpha}}{p_0\Gamma(\alpha)\Gamma(1-\alpha)},
 \qquad
 \kappa_n:=p_0\Gamma(\alpha)n^\alpha
 \sim\frac{1}{\Gamma(1-\alpha)\PP(\mathcal L_1>n)}.
\]
Thus the chain is $\beta$-null recurrent with $\beta=\alpha$, and
$\kappa_n$ is the return-count normalization, in the terminology of
\citet{KarlsenTjostheim2001}.

Define the cycle-sum processes
\[
 \bar{\mathcal L}_n(u)=\frac{1}{n}\sum_{i\leq\lfloor\kappa_nu\rfloor}\mathcal L_i,
 \qquad
 \bar{\mathcal C}_n(u)=\frac{1}{\sqrt{\kappa_n}}
        \sum_{i\leq\lfloor\kappa_nu\rfloor}\mathcal C_i.
\]
The length-tail asymptotics give $\bar{\mathcal L}_n\Rightarrow D_\alpha$, where
$D_\alpha$ is the stable subordinator normalized by
$\E e^{-sD_\alpha(u)}=e^{-us^\alpha}$.
By Donsker's theorem and eq.~\eqref{eq:cycle-isometry},
$\bar{\mathcal C}_n\Rightarrow\sqrt{v(a)}B_\alpha$, where $B_\alpha$ is standard
Brownian motion. For iid bivariate sums, \citet[Theorem~3]{ResnickGreenwood1979}
shows that these marginal limits imply joint convergence in
$D([0,\infty))^2$ with the product Skorokhod $J_1$ topology, and that
$D_\alpha$ and $B_\alpha$ are independent. No independence of
$\mathcal L_i$ and $\mathcal C_i$ within a cycle is required.

The stable subordinator is almost surely strictly increasing and
unbounded, so inversion at level one is continuous at its paths.
Since
\[
 \inf\{u:\bar{\mathcal L}_n(u)>1\}=\frac{N_n+1}{\kappa_n},
\]

we recover the occupation limit $N_n/\kappa_n\Rightarrow\mathcal M_\alpha$
of \citet[Theorem~2]{Pakes1971}, now realizing the
Mittag--Leffler variable from Section~\ref{sec:strict} as
\[
 \mathcal M_\alpha=\inf\{u:D_\alpha(u)>1\}.
\]
The joint limit below requires this representation.

The Brownian limit is continuous, so joint convergence and evaluation
of $\bar{\mathcal C}_n$ at $N_n/\kappa_n$ give
\[
 \left(\frac{N_n}{\kappa_n},
       \frac{\sum_{i\leq N_n}\mathcal C_i}{\sqrt{\kappa_n}}\right)
 \Rightarrow
 \left(\mathcal M_\alpha,\sqrt{v(a)}B_\alpha(\mathcal M_\alpha)\right).
\]
Here $B_\alpha$ is independent of $\mathcal M_\alpha$, since the latter is a
function of $D_\alpha$.

To pass from completed cycles to $M_n(a)$, let $\mathcal C_i^*$ be the
maximum absolute partial score within cycle $i$. Doob's inequality
and eq.~\eqref{eq:cycle-isometry} give $\E(\mathcal C_i^*)^2\leq4v(a)$.
For fixed $C,\eta>0$,
\[
 \PP\left(\max_{i\leq\lfloor C\kappa_n\rfloor+1}
          \mathcal C_i^*>\eta\sqrt{\kappa_n}\right)
 \leq(\lfloor C\kappa_n\rfloor+1)
       \PP(\mathcal C_1^*>\eta\sqrt{\kappa_n})\longrightarrow0.
\]
Together with tightness of $N_n/\kappa_n$, this makes the unfinished
cycle negligible. Eq.~\eqref{eq:occupation-ratios} now gives
eq.~\eqref{eq:strict-joint}, since
$(B_\alpha(\mathcal M_\alpha),\mathcal M_\alpha)\stackrel d=
(\sqrt{\mathcal M_\alpha}\mathcal Z,\mathcal M_\alpha)$ with $\mathcal Z$ independent of
$\mathcal M_\alpha$. The joint limit gives
$H_n^{-1}=\Op(n^{-\alpha})$. Together with the slope bound in
Appendix~\ref{app:slope-reduction}, this makes the slope contribution
in eq.~\eqref{eq:reduction} of order
$\Op(n^{-\alpha})=\op(n^{-\alpha/2})$, proving eq.~\eqref{eq:strict-fixed}.
Residual studentization is justified in
Appendix~\ref{app:residual-studentization}.
\end{proof}

\subsection{The joint limit at the recurrence boundary}
\label{app:boundary-potential}

Throughout this subsection, $\tau=1$. We prove the joint limit
in eq.~\eqref{eq:boundary-joint} first for $a=1$.
The weight $(1+x)^{-1}$ is the derivative of $\log(1+x)$,
suggesting a connection between the weighted innovations $M_n(1)$ and
fluctuations of the logarithm of the population size. To obtain an exact
martingale decomposition, we use a function $g$ that grows like $\log x$
and has conditional drift only at zero. Lemma~\ref{lem:potential}
establishes its properties, and its martingale decomposition gives a
reflected process. Lemma~\ref{lem:operational-time} proves its Brownian
limit using quadratic variation as time, and relates that time to
observation time. We then recover the limit at observation $n$ and
transfer it to the WLS score and accumulated weight.

In this construction only, write $H_n=H_n(1)$,
$M_n=M_n(1)$, and $Q_n=Q_n(1)$. Appendix~\ref{app:boundary-offsets} extends the
result to every fixed $a>0$.

\paragraph{The logarithmic potential.}

Start the process at zero.
Write $F$ and $G$ for the offspring and immigration probability generating
functions, and $F_j$ for the $j$th iterate of $F$, with $F_0(s)=s$.
The return probabilities satisfy
\[
 u_j=\prod_{i=0}^{j-1}G(F_i(0)),\qquad u_0=1.
\]
At $\tau=1$, \citet[Lemma~2.10]{WeiWinnicki1989} and the
critical offspring extinction estimate of
\citet[Corollary~1]{KestenNeySpitzer1966} give
\begin{equation}
 u_j\sim\frac{p_0}j,\qquad
 \omega_j:=-\log F_j(0)\sim\frac{2}{\sigma^2j},\qquad p_0>0.
 \label{eq:local-inputs}
\end{equation}
For nonnegative integers $x$, define the potential
\[
 g(x)=\frac1{p_0}\sum_{j=0}^{\infty}
       \{\PP_0(X_j=0)-\PP_x(X_j=0)\}.
\]
Using $\PP_x(X_j=0)=u_jF_j(0)^x$ gives the explicit form
\begin{equation}
 g(0)=0,\qquad
 g(x)=\frac1{p_0}\left[1+\sum_{j=1}^\infty u_j(1-e^{-\omega_jx})\right]
 \quad(x>0).
 \label{eq:potential}
\end{equation}
The series is finite for every finite $x$, since its tail is $O(x/j^2)$.

Write $Pg(x)=\E_xg(X_1)$ for the transition operator applied to $g$.

\begin{lemma}[Logarithmic potential]
\label{lem:potential}
The function $g$ in eq.~\eqref{eq:potential} satisfies
$g(0)=0\leq g(x)$ and
\begin{align}
 Pg(x)-g(x)&=p_0^{-1}\1_{\{x=0\}},\label{eq:poisson}\\
 g(x)/\log x&\to1\quad(x\to\infty).\label{eq:g-asymp}
\end{align}
Let $\zeta=g(X_1)-Pg(x)$ under $\PP_x$, and $v_g(x)=\E_x\zeta^2$.
Then, as $x\to\infty$,
\begin{align}
 \E_x\left[\zeta-\frac{X_1-x-\mu}{1+x}\right]^2&=o(x^{-1}),
 \label{eq:score-approx}\\
 x\,v_g(x)&\to\sigma^2.\label{eq:vg}
\end{align}
For some $p>2$ in the range of eq.~\eqref{eq:innovation-moment} and positive
finite constants $c_g,C_g,C$, the following bounds hold uniformly over
states in the communicating class containing zero:
\begin{equation}
 \frac{c_g}{1+x}\leq v_g(x)\leq\frac{C_g}{1+x},\qquad
 \E_x|\zeta|^p\leq C v_g(x).
 \label{eq:global-bounds}
\end{equation}
\end{lemma}

\begin{proof}
For the partial potential sum through $j=N$ use
$\PP_x(X_j=0)=u_jF_j(0)^x$. Its one-step drift telescopes to
$p_0^{-1}[\1_{\{x=0\}}-\PP_x(X_{N+1}=0)]$.
Monotone convergence for the potential and finiteness of $\E_xg(X_1)$
allow passage to the limit and prove eq.~\eqref{eq:poisson}.
Finiteness follows, for example, from $g(y)\leq C[1+\log(1+y)]$,
which the series and eq.~\eqref{eq:local-inputs} imply.

Extend eq.~\eqref{eq:potential} smoothly to $x>0$.
Differentiation gives
\[
 g'(x)=p_0^{-1}\sum_{j\geq1}u_j\omega_je^{-\omega_jx}.
\]
Splitting the sum at $j=\eta x$ and $j=x/\eta$ and taking a Riemann
sum on the middle interval gives
\[
 xg'(x)\to\int_0^\infty \frac{2}{\sigma^2z^2}e^{-2/(\sigma^2z)}\,dz=1.
\]%
The small-$j$ tails are exponentially suppressed and the large-$j$ tails
are bounded by summable inverse powers. Integrating the derivative
asymptotic gives $g(x)/\log x\to1$. These bounds also give
$g'(x)\leq C/(1+x)$ for $x\geq1$.

We next establish the probabilistic bounds needed for linearization.
Choose $p$ as in eq.~\eqref{eq:innovation-moment} and put $\Delta=X_1-x$.
Since $\Delta=W_1+\mu$ under $\PP_x$, that bound gives
$\E_x|\Delta|^p\leq C(1+x^{p/2})$.
For each fixed $\eta\in(0,1)$,
$\PP_x(X_1<(1-\eta)x)\leq e^{-c_\eta x}$ for sufficiently large $x$,
because immigration is nonnegative and the offspring lower tail has a Chernoff
bound using $\E e^{-tZ}<\infty$ and $\E Z=1$.
Define $\rho_x=g(X_1)-g(x)-\Delta/x$. On
$|\Delta|\leq\eta x$, with $0<\eta<1/2$, the asymptotic
$xg'(x)\to1$ and the mean-value theorem give
\[
 g(X_1)-g(x)=\Delta/x+\rho_x,\qquad
 \E_x[\rho_x^2;|\Delta|\leq\eta x]
 \leq\frac{o(1)+O(\eta^2)}{x}.
\]
Indeed, uniformly for $\xi\in[(1-\eta)x,(1+\eta)x]$,
$xg'(\xi)=1+O(\eta)+o(1)$, and $\E_x\Delta^2=O(x)$.
On the upper complement $\{\Delta>\eta x\}$, the derivative
bound gives $|g(X_1)-g(x)|\leq C\Delta/x$, and the $p$th-moment bound
gives
\[
 x\E_x[(\Delta/x)^2;\Delta>\eta x]
 \leq C\eta^{2-p}x^{1-p/2}\longrightarrow0.
\]
On the lower complement $\{\Delta<-\eta x\}$, the bound
$|g(X_1)-g(x)|\leq C\log(1+x)$ and the exponential bound give a
contribution of at most $Cx\log^2(1+x)e^{-c_\eta x}\to0$.
Replacing $\Delta/x$ by $(\Delta-\mu)/(1+x)$ has mean-square cost
$O(x^{-2})$, and $Pg(x)=g(x)$ for $x>0$ by eq.~\eqref{eq:poisson},
so that $\zeta=g(X_1)-g(x)$ there.
Let $x\to\infty$ and then $\eta\downarrow0$ to get
eq.~\eqref{eq:score-approx}. Since
$\E_x[(X_1-x-\mu)/(1+x)]^2\sim\sigma^2/x$, it also gives
eq.~\eqref{eq:vg}.

For the $p$th moment, split at $X_1=x/2$.
On the upper part the derivative bound controls the increment by
$C|\Delta|/x$; on the lower part use the exponential bound.
Thus $\E_x|\zeta|^p\leq Cx^{-p/2}$ for large $x$.
The finitely many remaining states in that class have finite moments and
$v_g(x)>0$ because the transition is nondegenerate and $g$ is strictly
increasing, as is clear from eq.~\eqref{eq:potential}. Enlarging constants proves eq.~\eqref{eq:global-bounds}.
\end{proof}

\paragraph{Reflection and the change of time.}
By eq.~\eqref{eq:poisson}, the process $Z_n^g$ defined below is a
martingale, with predictable quadratic variation as displayed:
\begin{equation}
 \begin{aligned}
 Z_n^g&=g(X_n)-L_n,\qquad
 L_n=p_0^{-1}\sum_{k=0}^{n-1}\1_{\{X_k=0\}},\\
 \langle Z^g\rangle_n&=\sum_{k=0}^{n-1}v_g(X_k).
 \end{aligned}
 \label{eq:decomposition}
\end{equation}
The increasing process $L_n$ grows by $p_0^{-1}$ only if $g(X_{n-1})=0$.
Nonnegativity and this support property imply the deterministic bound
\begin{equation}
 0\leq L_n-\max_{k\leq n}(-Z_k^g)_+\leq p_0^{-1}.
 \label{eq:skorokhod-error}
\end{equation}
Thus the decomposition is a discrete reflection with a uniformly bounded
error. We call $\langle Z^g\rangle_n$, the predictable quadratic variation of the martingale
$Z^g$, its intrinsic time, and $n$ observation time.

\begin{lemma}[Intrinsic time and observation time]
\label{lem:operational-time}
Let $r\to\infty$ and $k_r(s)=\inf\{k:\langle Z^g\rangle_k>r^2s\}$.
Use the standard Brownian motion $B$, its regulator $\ell$, and the
reflected process $R$ from Section~\ref{sec:boundary}.
Then on compact intrinsic-time intervals,
\begin{equation}
 \left(\frac{Z^g_{k_r(s)}}r,\frac{L_{k_r(s)}}r,
                  \frac{g(X_{k_r(s)})}r\right)
 \Rightarrow (B(s),\ell(s),R(s)).
 \label{eq:operational}
\end{equation}
Jointly, for every fixed $s>0$,
\begin{equation}
 \frac{\log(1+k_r(s))}{r}
 \Rightarrow \sup_{u\leq s}R(u).
 \label{eq:physical}
\end{equation}
The convergence holds jointly for finitely many such times, with the
same limiting path as in eq.~\eqref{eq:operational}.
\end{lemma}

\begin{proof}
\textit{Step 1: The reflected Brownian limit in intrinsic time.}
The chain is recurrent and $v_g(0)>0$, so $\langle Z^g\rangle_n\to\infty$.
The inverse is an integer-valued stopping time, since
$\{k_r(s)\leq m\}=\{\langle Z^g\rangle_m>r^2s\}\in\F_{m-1}$ for $m\geq1$.
Since the strictly positive increments of $\langle Z^g\rangle$ are bounded,
\[
 r^2s<\langle Z^g\rangle_{k_r(s)}\leq r^2s+C.
\]
On each fixed intrinsic-time horizon $s_0>0$, optional sampling is justified in
$L^2$ by this bound on the expected stopped bracket. With filtration
$\F_{k_r(s)}$, the centered intrinsic-time process
$(Z^g_{k_r(s)}-Z^g_1)/r$ starts at zero and is a square-integrable
martingale. Its bracket is $(\langle Z^g\rangle_{k_r(s)}-\langle Z^g\rangle_1)/r^2$, because its jump times
$\langle Z^g\rangle_{j-1}/r^2$, $j\geq2$, and conditional jump variances are predictable
in this filtration.
 This bracket converges uniformly to $s$.
The conditional Lindeberg sum up to $k_r(s_0)$ is bounded in expectation by
\[
 \frac{C}{\eta^{p-2}r^p}\E \langle Z^g\rangle_{k_r(s_0)}=O(r^{2-p}).
\]
Writing $\Delta_r^g=\max_{j\leq k_r(s_0)}|Z_j^g-Z_{j-1}^g|$, the same bound gives
\begin{equation}
 \E (\Delta_r^g)^p\leq C\E \langle Z^g\rangle_{k_r(s_0)}=O(r^2),\qquad
 \E(\Delta_r^g/r)^2=O(r^{4/p-2})\longrightarrow0.
 \label{eq:max-jump}
\end{equation}
The predictable bracket jumps are bounded by $C/r^2$, so the martingale
functional CLT \citep[Theorem~2.1(ii)]{Whitt2007} applies. The original
intrinsic-time process has initial value $Z_1^g/r\to0$ in $L^2$, since
$k_r(0)=1$. Restoring this value gives the Brownian coordinate.
The Lipschitz reflection map and eq.~\eqref{eq:skorokhod-error} give the other
two coordinates. By eq.~\eqref{eq:max-jump}, one-observation errors in
$Z^g/r$ vanish on localized intrinsic-time horizons; increments of $L/r$
are at most $1/(p_0r)$ and the same conclusion holds for $g(X)/r$.
All these limits have continuous paths, so a Skorokhod representation
gives uniform convergence on compact intervals.

\textit{Step 2: An upper bound on observation time.}
We verify eq.~\eqref{eq:physical} on that representation. Fix $s>0$ and
write $\overline R_s=\sup_{0\leq u\leq s}R(u)$, which is positive almost surely.
Fix $0<\eta<1/2$. By eq.~\eqref{eq:g-asymp}, choose $C_\eta<\infty$ so that, for all states,
\[
 (1-\eta)\log(1+x)-C_\eta\leq g(x)
 \leq(1+\eta)\log(1+x)+C_\eta.
\]
From the lower bound in eq.~\eqref{eq:global-bounds},
\[
 k_r(s)\leq c_g^{-1}(r^2s+C)
                \max_{j<k_r(s)}(1+X_j).
\]
Every integer $1\leq j\leq k_r(s)$ occurs as $k_r(u)$ for some
$0\leq u\leq s$, since $\langle Z^g\rangle$ has strictly positive increments.
Uniform convergence in eq.~\eqref{eq:operational}, including the fixed
initial state separately, implies
$\max_{j<k_r(s)}g(X_j)/r\leq \overline R_s+\eta$ for all sufficiently large $r$.
Consequently,
\[
 \frac{\log(1+k_r(s))}{r}
 \leq\frac{\overline R_s+\eta}{1-\eta}
   +\frac{C_\eta}{r(1-\eta)}
   +\frac{\log\{1+c_g^{-1}(r^2s+C)\}}{r}.
\]
This gives $\limsup_r\log(1+k_r(s))/r\leq \overline R_s$ by first taking
$r\to\infty$ and then $\eta\downarrow0$.

\textit{Step 3: A lower bound on observation time.}
Also take $\eta<\overline R_s$. By continuity choose
$0\leq s_1<s_2\leq s$ such that
$R(u)>\overline R_s-\eta/2$ throughout $[s_1,s_2]$; a maximum at $s$ is
handled by a left interval. On the representation, uniform error
less than $\eta/2$ gives
$g(X_{k_r(u)})\geq r\Lambda_\eta$ on this interval for all sufficiently
large $r$, where $\Lambda_\eta=\overline R_s-\eta>0$.
For each integer $k_r(s_1)\leq j<k_r(s_2)$ there is
$u\in[s_1,s_2]$ with $k_r(u)=j$, namely $u=s_1$ for the first integer
and $u=\langle Z^g\rangle_{j-1}/r^2$ for the others. Thus the global upper bound for $g$
gives the explicit population estimate
\[
 \min_{k_r(s_1)\leq j<k_r(s_2)}(1+X_j)
 \geq\exp\left\{\frac{r\Lambda_\eta-C_\eta}{1+\eta}\right\}.
\]
The bracket overshoot bound and the upper variance bound now imply
\[
 r^2(s_2-s_1)-2C
 \leq \sum_{j=k_r(s_1)}^{k_r(s_2)-1}v_g(X_j)
 \leq C_g[k_r(s_2)-k_r(s_1)]
       \exp\left\{-\frac{r\Lambda_\eta-C_\eta}{1+\eta}\right\}.
\]
For $r$ large enough that the left side is positive, taking logarithms yields
\[
 \frac{\log(1+k_r(s))}{r}
 \geq\frac{\Lambda_\eta}{1+\eta}
  +\frac{\log\{r^2(s_2-s_1)-2C\}-\log C_g-C_\eta/(1+\eta)}{r}.
\]
For each fixed $\eta$ and represented path, $s_2-s_1>0$ is fixed,
so the second term tends to zero. Letting $\eta\downarrow0$ gives
$\liminf_r\log(1+k_r(s))/r\geq \overline R_s$.
Applying this argument at any finite set of times proves the
stated joint convergence. Monotonicity supplies the bounds needed for
inversion at a fixed level.
\end{proof}

\begin{proof}[Proof of Theorem~\ref{thm:boundary}, score and clock for $a=1$]
\textit{Step 1: Returning to observation time.}
Apply eq.~\eqref{eq:physical} with $r=\log n$. A reflected Brownian path
first reaches level one at the finite time $T$ and exceeds one
arbitrarily soon afterwards almost surely. Before $T$ its running
maximum is less than one. The monotone inverse bounds in
eq.~\eqref{eq:physical} therefore give, jointly with the intrinsic-time path,
\[
 \langle Z^g\rangle_n/(\log n)^2\Rightarrow T,\qquad
 (Z_n^g,L_n)/\log n\Rightarrow (B(T),\ell(T)).
\]
More precisely, sandwich $\langle Z^g\rangle_n/r^2$ between rational
intrinsic times whose running maxima are below and above one, and let
both times tend to $T$.
Continuous-path convergence in
eq.~\eqref{eq:operational} permits evaluation at the limiting inverse.
Since $\langle Z^g\rangle$ is strictly increasing,
$k_r(\langle Z^g\rangle_n/r^2)=n+1$, a one-observation indexing offset.
 Localize to
$\{\langle Z^g\rangle_n/r^2\leq s_0\}$ and apply eq.~\eqref{eq:max-jump} to replace
$Z^g_{n+1}/r$ by $Z^g_n/r$, and similarly for $L$ and $g(X)$.
Tightness of $\langle Z^g\rangle_n/r^2$ removes the localization by letting $s_0\to\infty$.
Since $R(T)=1$, $B(T)=1-U$.

\textit{Step 2: Transferring the limit to the WLS score and clock.}
By \citet[Remark~2.23, equation~(2.43)]{WeiWinnicki1989},
$\pi(w_1)=\infty$ at the boundary, whereas $\pi$ is finite on finite
sets. With $N_n$ counting completed cycles, truncation and the cycle
strong law give $H_n/N_n\to\infty$ almost surely. For any finite set
$A$, the same strong law gives
$\sum_{t=1}^n\1_A(X_{t-1})/N_n\to\pi(A)$ almost surely; the
unfinished integrable cycle reward is $o(N_n)$ almost surely.
Thus every finite-set occupation is $o(H_n)$ almost surely.
Eqs.~\eqref{eq:score-approx} and~\eqref{eq:vg}, first outside a finite set and then inside,
imply
\[
 \begin{gathered}
 \langle Z^g\rangle_n/H_n\to\sigma^2,\qquad Q_n/H_n\to\sigma^2,\\
 \langle Z^g-M\rangle_n/H_n\to0
 \quad\text{almost surely}.
 \end{gathered}
\]
The first limit and the tight limit for $\langle Z^g\rangle_n/(\log n)^2$ show that
the bracket of $Z^g-M$, divided by $(\log n)^2$, tends to zero
in probability. The stopped martingale maximal inequality then gives
$(Z_n^g-M_n)/\log n\to_p0$. This proves
the score and clock assertions in eq.~\eqref{eq:boundary-joint} for $a=1$.
The clock limit gives $H_n^{-1}=\Op((\log n)^{-2})$.
Together with the slope bound in Appendix~\ref{app:slope-reduction},
this makes the slope contribution in eq.~\eqref{eq:reduction} of order
$\Op((\log n)^{-2})$, which is negligible after multiplication by
$\log n$. Thus eq.~\eqref{eq:reduction} gives eq.~\eqref{eq:boundary-est} for
$a=1$. Residual studentization is established in
Appendix~\ref{app:residual-studentization}.
\end{proof}

\begin{remark}[Expected accumulated weight]
\label{rem:boundary-mean}
The sharp expectation $\E H_n(a)\sim(\log n)^2/\sigma^2$ follows
from the positive-real estimate of
\citet[Lemma~2.2, equation~(2.4)]{LiZhang2021}. Their lemma is stated
under an additional offspring lattice-span condition. For the
positive-real estimate needed here, the following argument avoids
that condition.

For any fixed $a>0$, start at zero and choose $j=j(s)$ so that
$F_j(0)\leq e^{-s}\leq F_{j+1}(0)$. Monotonicity of the generating
functions gives
\[
 \frac{u_{n+j}}{u_j}
 \leq\E_0e^{-sX_n}
 \leq\frac{u_{n+j+1}}{u_{j+1}}.
\]
As $s\downarrow0$, eq.~\eqref{eq:local-inputs} gives
$j(s)\sim2/(\sigma^2s)$. Both ratios are therefore
$(1+\sigma^2ns/2)^{-1}\{1+o(1)\}$, uniformly for
$0<s\leq\eta_n$ whenever $\eta_n\downarrow0$, because the smallest possible
$j(s)$ then tends to infinity and the asymptotics for $u_j$ and $F_j(0)$
hold uniformly on this range.
 Taking $\eta_n=(\log n)^{-1/2}$,
integrating against $e^{-as}$ over $(0,\eta_n)$, and
bounding the remaining integral by
$a^{-1}\E_0e^{-\eta_nX_n}=O((n\eta_n)^{-1})$ gives
\[
 \E_0(a+X_n)^{-1}
 \sim\frac{2\log(1+\sigma^2n\eta_n/2)}{\sigma^2n}
 \sim\frac{2\log n}{\sigma^2n}.
\]
An independent, almost surely finite initial population multiplies
the Laplace transform by $\E[F_n(e^{-s})^{X_0}]$.
This factor tends to one uniformly in $s\geq0$, since $F_n(0)\uparrow1$.
Summing over generations yields
$\E H_n(a)\sim(\log n)^2/\sigma^2$.
\end{remark}

\subsection{Transferring the boundary limit to fixed offsets}
\label{app:boundary-offsets}

For fixed $a>0$,
$|w_a-w_1|+|w_a^2V-\sigma^2w_1|\leq C_aw_1^2$ and
$(w_a-w_1)^2V\leq C_aw_1^3$.
Eq.~\eqref{eq:harmonic-bound}, Markov's inequality and the
martingale isometry therefore give
\begin{align}
 H_n(a)-H_n(1)&=\Op(\log n),\qquad
 Q_n(a)-\sigma^2H_n(1)=\Op(\log n),
 \label{eq:boundary-offset-clocks}
 \\
 M_n(a)-M_n(1)&=\Op(\sqrt{\log n}).
 \label{eq:boundary-offset-scores}
\end{align}
Together with the $a=1$ clock limit, these bounds transfer the joint
score and clock limit to every fixed offset.

For each fixed $a>0$, the transferred clock limit gives
$H_n(a)^{-1}=\Op((\log n)^{-2})$. The slope bound in
Appendix~\ref{app:slope-reduction} therefore makes the slope contribution
in eq.~\eqref{eq:reduction} of order $\Op((\log n)^{-2})$. The identity
\[
 \frac{M_n(a)}{H_n(a)}-\frac{M_n(1)}{H_n(1)}
 =\frac{M_n(a)-M_n(1)}{H_n(a)}
 +\frac{M_n(1)}{H_n(1)}
       \left\{\frac{H_n(1)}{H_n(a)}-1\right\}
\]
is $\op((\log n)^{-1})$ by the clock and score comparisons and
positivity of $T$. This proves eq.~\eqref{eq:boundary-est} and
\begin{equation}
 \log n\{\widehat\mu_n(a)-\widehat\mu_n(1)\}\to_p0.
 \label{eq:boundary-offset-equivalence}
\end{equation}
Residual studentization is established in Appendix~\ref{app:residual-studentization}.

\subsection{Residual studentization}
\label{app:residual-studentization}

Fix $a>0$ and abbreviate $H_n=H_n(a)$ and $Q_n=Q_n(a)$,
suppressing the offset in the fitted coefficients and the residual
estimate $\widehat Q_n(a)$.
Define the squared-score sum
\begin{equation}
 \widetilde Q_n(a)=\sum_{t=1}^n w_a(X_{t-1})^2W_t^2.
 \label{eq:squared-score}
\end{equation}
Also define the martingale differences
$\Delta_t^Q=w_a(X_{t-1})^2\{W_t^2-V(X_{t-1})\}$, so that
$\widetilde Q_n(a)-Q_n(a)=\sum_{t=1}^n\Delta_t^Q$.
For both recurrent regimes, we will establish
\begin{equation}
 \frac{\widetilde Q_n(a)}{Q_n(a)}\to_p1,\qquad
 \frac{\widehat Q_n(a)}{Q_n(a)}\to_p1.
 \label{eq:fixed-score-consistency}
\end{equation}
The first assertion concerns the innovations; the second replaces
them by fitted residuals.

\paragraph{Strict recurrence.}
First suppose $0<\tau<1$.
Let
\[
 \mathcal U_i=\sum_{t=\mathcal S_{i-1}+1}^{\mathcal S_i}w_a(X_{t-1})^2W_t^2
\]
be the squared-score reward over cycle $i$. Conditional expectation and
Tonelli's theorem give
\[
 \E\mathcal U_i
 =\E\sum_{t=\mathcal S_{i-1}+1}^{\mathcal S_i}w_a(X_{t-1})^2V(X_{t-1})
 =v(a)<\infty.
\]
The cycle strong law applies, and integrability gives
$\mathcal U_i/i\to0$ almost surely. The unfinished reward is bounded by
$\mathcal U_{N_n+1}=o(N_n)$. Since $H_n/N_n\to h(a)$, the total
squared-score reward divided by $H_n$ tends to $v(a)/h(a)$.
Since $Q_n/H_n\to v(a)/h(a)>0$, this proves $\widetilde Q_n(a)/Q_n(a)\to1$
almost surely.

\paragraph{The recurrence boundary.}
Suppose $\tau=1$. Choose $p$ as in eq.~\eqref{eq:innovation-moment}
and put $q=p/2>1$.
The moment bound gives
$\E(|\Delta_t^Q|^q\mid\F_{t-1})\leq C_a(1+X_{t-1})^{-q}$.
Since $1<q\leq3/2$, the martingale $L^q$ inequality and
eq.~\eqref{eq:harmonic-bound} give directly
\[
 \E\left|\sum_{t=1}^n\Delta_t^Q\right|^q
 \leq C_a\sum_{t=1}^n\E(1+X_{t-1})^{-q}=O(\log n).
\]
Hence $(\log n)^{-2}\sum_{t=1}^n\Delta_t^Q\to_p0$.
The positive limit of $H_n/(\log n)^2$ permits division by $H_n$.
Since $Q_n/H_n\to\sigma^2$, this proves $\widetilde Q_n(a)/Q_n(a)\to_p1$.

\paragraph{Replacing innovations by fitted residuals.}
In either recurrent regime, the coefficient rates already proved give
\begin{equation}
 \begin{aligned}
 \mathcal E_n(a)
 &:=\sum_{t=1}^n w_a(X_{t-1})^2
       \{(\widehat m_n-1)X_{t-1}+\widehat\mu_n-\mu\}^2\\
 &\leq2n(\widehat m_n-1)^2
       +\frac2a(\widehat\mu_n-\mu)^2H_n
 =\Op(1)=\op(H_n).
 \end{aligned}
 \label{eq:residual-perturbation}
\end{equation}
By Cauchy--Schwarz,
\[
 |\widehat Q_n-\widetilde Q_n(a)|
 \leq2\sqrt{\widetilde Q_n(a)\mathcal E_n(a)}+\mathcal E_n(a)=\op(H_n).
\]
Here $\widetilde Q_n(a)/H_n=\Op(1)$ by the preceding arguments, and
$Q_n\asymp H_n$. This proves the second assertion of
eq.~\eqref{eq:fixed-score-consistency}. Combining it with the joint limits
and the negligible slope contribution proves eqs.~\eqref{eq:strict-pivot}
and~\eqref{eq:boundary-pivot}.

\begin{proof}[Proof of Proposition~\ref{prop:transient-fixed} and the transient case of Corollary~\ref{cor:maintained}]

For $\tau>1$, put $x_t=X_{t-1}$ and
$c=\mu-\sigma^2/2>0$. By \citet[Lemma~2.4]{WeiWinnicki1990},
$H_n(1)/\log n\to_p c^{-1}$ and
$\sum_{t\geq1}(1+x_t)^{-q}<\infty$ almost surely for every $q>1$.

We establish the following bounds uniformly over compact offset
intervals for later use with estimated offsets in
Appendix~\ref{ad:transient-proofs}.
For a fixed compact interval $I\subset(0,\infty)$, uniformly in
$a\in I$ and $x\geq0$,
\[
 |w_a(x)-w_1(x)|\leq C_I(1+x)^{-2},\qquad
 |w_a(x)^2-w_1(x)^2|\leq C_I(1+x)^{-3}.
\]
Since $\E(|W_t|\mid\F_{t-1})\leq C(1+x_t)^{1/2}$ and
$\E(W_t^2\mid\F_{t-1})\leq C(1+x_t)$, conditional summability gives
\[
 \sum_{t\geq1}\frac{|W_t|}{(1+x_t)^2}<\infty,\qquad
 \sum_{t\geq1}\frac{W_t^2}{(1+x_t)^3}<\infty
 \quad\text{almost surely}.
\]
Consequently, almost surely,
\begin{equation}
 \begin{aligned}
 \sup_{n\geq1,\,a\in I}
 \bigl\{|H_n(a)-H_n(1)|+|M_n(a)-M_n(1)|\bigr\}&<\infty,\\
 \sup_{n\geq1,\,a\in I}
 \left\{|\widetilde Q_n(a)-\widetilde Q_n(1)|
       +\sum_{t=1}^n w_a(x_t)^2\right\}&<\infty.
 \end{aligned}
 \label{eq:transient-offset-bounds}
\end{equation}
These comparisons use only the original moment assumption.

The known drift limit for $a=1$, the slope bound
$\widehat m_n(1)-1=\Op(n^{-1})$ from
\citet[Theorem~2.5]{WeiWinnicki1990}, and eq.~\eqref{eq:reduction} give
$M_n(1)=\Op(\sqrt{\log n})$.
For each fixed $a>0$, eq.~\eqref{eq:transient-offset-bounds} then gives
$H_n(a)/\log n\to_p c^{-1}$ and
$M_n(a)=\Op(\sqrt{\log n})$.
Eqs.~\eqref{eq:critical-functional-inputs} and~\eqref{eq:slope-normal}
yield $\widehat m_n(a)-1=\Op(n^{-1})$.
The score ratios $M_n(a)/H_n(a)$ and $M_n(1)/H_n(1)$ differ by
$\Op((\log n)^{-1})$, as do the slope contributions in
eq.~\eqref{eq:reduction}. Hence
$\widehat\mu_n(a)-\widehat\mu_n(1)=\Op((\log n)^{-1})$, so that
\begin{equation}
 \sqrt{\log n}\{\widehat\mu_n(a)-\widehat\mu_n(1)\}\to_p0.
 \label{eq:transient-offset-equivalence}
\end{equation}
The constrained estimator differs from the joint estimator by
$\Op((\log n)^{-1})$ as well. Combining these bounds with the known
drift limit for $a=1$ proves eq.~\eqref{eq:transient} for both fits.

Fix $a>0$ and abbreviate $H_n=H_n(a)$ and $Q_n=Q_n(a)$.
Here $H_n\to\infty$ almost surely by comparison with $H_n(1)$.
Since
$Q_n=\sigma^2H_n+(b-\sigma^2a)\sum_{t=1}^n(a+X_{t-1})^{-2}$,
we have $Q_n/H_n\to\sigma^2$ almost surely.
Choose $q=1+\min\{\delta/2,1/2\}$, where $\delta$ is from
Assumption~\ref{ass:model}. Eq.~\eqref{eq:innovation-moment},
with $p=2q$, yields
$\sum_{t\geq1}\E(|\Delta_t^Q|^q\mid\F_{t-1})<\infty$ almost surely,
so the martingale $\sum_{t=1}^n\Delta_t^Q$ converges almost surely.
Consequently,
$H_n^{-1}\sum_{t=1}^n w_a(X_{t-1})^2W_t^2\to\sigma^2$ almost surely.
The slope bound $n(\widehat m_n(a)-1)=\Op(1)$ and the drift limit above
make the bound in eq.~\eqref{eq:residual-perturbation} $\Op(1)=\op(H_n)$.
Cauchy--Schwarz therefore gives $\widehat Q_n(a)/H_n\to_p\sigma^2$,
and Slutsky's theorem proves eq.~\eqref{eq:fixed-transient-pivot}.
For the constrained fit, set the slope error to zero and replace the
drift error by $\widehat\mu_{n,\mathrm c}(a)-\mu$ in
eq.~\eqref{eq:residual-perturbation}. Its rate gives the same bound, so
$\widehat Q_{n,\mathrm c}(a)/H_n\to_p\sigma^2$ and the constrained studentized limit follows.
Since $Q_n/H_n\to\sigma^2$ almost surely, this also gives
$\widehat Q_{n,\mathrm c}(a)/Q_n(a)\to_p1$.

\end{proof}

\begin{proof}[Proof of Corollary~\ref{cor:maintained}]
For $0<\tau\leq1$, the exact score ratio in
eq.~\eqref{eq:constrained-error} and the
joint limits give the drift-estimator distributions. Put
$d_{\mathrm c}=\widehat\mu_{n,\mathrm c}(a)-\mu$. The score and clock scales imply
$d_{\mathrm c}^2H_n(a)=\Op(1)$ in both recurrent regimes. Since $w_a^2\leq w_a/a$,
the squared residual perturbation is at most $d_{\mathrm c}^2H_n(a)/a=\Op(1)$.
Eq.~\eqref{eq:fixed-score-consistency} and Cauchy--Schwarz
therefore give $\widehat Q_{n,\mathrm c}(a)/Q_n(a)\to_p1$ and both studentized limits.

For $\tau>1$, the preceding proof of
Proposition~\ref{prop:transient-fixed} establishes the constrained limit
and residual consistency.

\end{proof}

\subsection{Boundary distribution and numerical quantiles}
\label{app:boundary-quantiles}

Throughout this subsection, $\tau=1$, and $(T,U)$ is the stopped
Brownian pair defined in Section~\ref{sec:boundary}.

\paragraph{Joint transform.}
The hitting time satisfies $0<T<\infty$ almost surely. Write $\E_x^{R}$
for expectation under reflected Brownian motion started at $x$.
To compute the joint Laplace transform of $(T,U)$, define
$\psi(x)=\E_x^{R} e^{-sT-\theta\ell(T)}$ for $x\in[0,1]$. This function solves
$\psi''/2=s\psi$, $\psi(1)=1$, and $\psi'(0)=\theta\psi(0)$.
Solving this boundary-value problem at zero gives
\begin{equation}
 \E e^{-sT-\theta U}
 =\left[\cosh\sqrt{2s}
  +\frac{\theta}{\sqrt{2s}}\sinh\sqrt{2s}\right]^{-1},
 \quad s,\theta\geq0,
 \label{eq:joint-transform}
\end{equation}
with the continuous interpretation at $s=0$ and $\ell$ the Skorokhod
regulator. Setting $s=0$ gives
\[
 \E e^{-\theta U}=\frac1{1+\theta},\qquad \theta\geq0,
\]
so $U\sim\operatorname{Exp}(1)$, hence $U>0$ almost surely.
The transform also gives $\E T=1$ and $\PP(1-U<0)=e^{-1}$.
This sign probability alone excludes a standard normal law for the
studentized limit.

\paragraph{Numerical quantiles.}
Inverting eq.~\eqref{eq:joint-transform} in $\theta$, with $k=\sqrt{2s}$,
gives
\[
 \E(e^{-sT}\mid U=u)
 =\frac{k}{\sinh k}\exp\{-u(k\coth k-1)\},\qquad U\sim\mathrm{Exp}(1).
\]
The product for $\sinh$ and the partial-fraction expansion for $\coth$
therefore yield, for $\lambda_j=\pi^2j^2/2$,
\[
 T\mid U=u\ \overset d=\sum_{j\geq1}G_j,\qquad
 \mathcal N_j\mid U=u\sim\mathrm{Poisson}(2u),\qquad
 G_j\mid(U,\mathcal N_j)\sim\mathrm{Gamma}(1+\mathcal N_j,\text{rate }\lambda_j),
\]
with conditional independence across $j$.
After $K$ terms we replace the remaining sum by its conditional mean
$(1+2u)\sum_{j>K}\lambda_j^{-1}$. The unconditional mean-square
error in $T$ is $5\sum_{j>K}\lambda_j^{-2}\leq20/(3\pi^4K^3)$.
We compute the quantiles from one million independent draws using $K=128$.
The Monte Carlo standard errors of the estimated $q_{0.025}$ and
$q_{0.975}$ are $0.0031$ and $0.0024$, respectively, estimated from
100 disjoint batches.

\section{Proofs for estimated offsets}
\label{ad:appendix}

For the recurrent results, we first bound
the score and clock uniformly over compact offset intervals.
Consistency places the estimated offset in such an interval with
probability tending to one. After proving Theorem~\ref{ad:limit},
we give the closed form of the variance regression. For transience, we first
isolate the absence of atoms in the variance-regression limit, then prove
Proposition~\ref{ad:transientoffset} by comparing offsets on compact
intervals.

\subsection{Uniform bounds and recurrent estimated offsets}
\label{ad:recurrent-proofs}

Write $x_t=X_{t-1}$ and retain $S_n$ from
eq.~\eqref{eq:critical-functional-inputs}. Unless specified otherwise, sums over
$t$ in this appendix run from $1$ to $n$. Let $r_n$ be the normalization
in Theorem~\ref{ad:limit}. The stochastic orders and normal equations
established in Appendix~\ref{app:original} apply to every deterministic offset.

\begin{lemma}[Uniform bounds on compact offset intervals]
\label{ad:uniformlemma}
Let $I=[a_-,a_+]\subset(0,\infty)$ be fixed. Under
Assumption~\ref{ass:model} and $0<\tau\leq1$,
\begin{align}
 \sup_{a\in I}|M_n(a)|+\sup_{a\in I}|M_n'(a)|
     &=\Op(r_n),\label{ad:uniformscore}\\
 \sup_{a\in I}H_n(a)+\sup_{a\in I}|H_n'(a)|
     &=\Op(r_n^2),\label{ad:uniformclock}\\
 \left\{r_n^{-2}\inf_{a\in I}H_n(a)\right\}^{-1}
     &=\Op(1),\label{ad:lowerclock}\\
 \sup_{a\in I}|\widehat m_n(a)-1|&=\Op(n^{-1}),\label{ad:uniformslope}\\
 \sup_{a\in I}|\widehat\mu_n(a)-\mu|
     &=\Op(r_n^{-1}).\label{ad:uniformintercept}
\end{align}
\end{lemma}

\begin{proof}
The functions $w_a(x)$ are uniformly comparable with $(1+x)^{-1}$ on
$I$. Moreover $w_a(x)^2\leq C_I(1+x)^{-1}$. The positive occupation
limits in eqs.~\eqref{eq:strict-joint} and~\eqref{eq:boundary-joint}
therefore prove eqs.~\eqref{ad:uniformclock}--\eqref{ad:lowerclock}.

For $k=0,1,2$, the score derivatives are
\[
 M^{(k)}_n(a)=(-1)^k k!\sum_{t=1}^n
       \frac{W_t}{(a+x_t)^{k+1}}.
\]
The martingale isometry, $V(x)/(a+x)^{2k+2}\leq C_I/(1+x)$,
and eq.~\eqref{eq:harmonic-bound} give
\begin{equation}
 \sup_{a\in I}\E\{M^{(k)}_n(a)^2\}
      \leq C_I\E H_n(1)=O(r_n^2),\qquad k=0,1,2.
 \label{ad:scoreisometry}
\end{equation}
For any continuously differentiable $\varphi$ on $I$,
\[
 \sup_I|\varphi|^2\leq2|\varphi(a_-)|^2+
                2|I|\int_I|\varphi'(a)|^2\,da.
\]
Applying it to $M_n$ and $M'_n$, then taking expectations and using
eq.~\eqref{ad:scoreisometry}, proves eq.~\eqref{ad:uniformscore}.

The normal equations in eqs.~\eqref{eq:slope-normal} and~\eqref{eq:reduction}
yield the uniform coefficient bounds as follows.
Uniformly on $I$, $n^2/H_n(a)=\Op(n^2/r_n^2)=\op(n^2)$
and $nM_n(a)/H_n(a)=\Op(n/r_n)=\op(n)$, since $r_n\to\infty$.
Together with eq.~\eqref{eq:critical-functional-inputs}, these bounds show
that the weighted design is nonsingular uniformly on $I$ with probability
tending to one and prove eq.~\eqref{ad:uniformslope}. Substitution into
eq.~\eqref{eq:reduction} proves eq.~\eqref{ad:uniformintercept}.
\end{proof}

\begin{proof}[Proof of Theorem~\ref{ad:limit}]
\textit{Step 1: Equivalence with the fixed-offset estimator.}
Choose a compact interval $I$ whose interior contains $a_0$.
Consistency gives $\PP(\widehat a_n\in I)\to1$.
The mean value theorem and Lemma~\ref{ad:uniformlemma} give
\begin{align*}
 M_n(\widehat a_n)-M_n(a_0)&=\op(r_n),\\
 H_n(\widehat a_n)-H_n(a_0)&=\op(r_n^2).
\end{align*}
The bound on the inverse clock in eq.~\eqref{ad:lowerclock}, applied to the
difference of the two score ratios, then yields
\[
 \frac{M_n(\widehat a_n)}{H_n(\widehat a_n)}
       -\frac{M_n(a_0)}{H_n(a_0)}=\op(r_n^{-1}).
\]
Both slope terms in eq.~\eqref{eq:reduction} are uniformly
$\Op(r_n^{-2})$, proving eq.~\eqref{ad:oracleequiv}.

\textit{Step 2: Residual normalization at the estimated offset.}
Use $\widetilde Q_n(a)$ from eq.~\eqref{eq:squared-score}.
Eq.~\eqref{eq:fixed-score-consistency} and $Q_n(a_0)=\Op(r_n^2)$,
together with $\E \widetilde Q_n(1)=\E Q_n(1)\leq C\E H_n(1)=O(r_n^2)$, give
\begin{equation}
 \frac{\widetilde Q_n(a_0)-Q_n(a_0)}{r_n^2}\to_p0,
 \qquad \widetilde Q_n(1)=\Op(r_n^2).
 \label{ad:squaredreward}
\end{equation}

On $I$,
\[
 \sup_{a\in I}|\widetilde Q_n'(a)|\leq C_I \widetilde Q_n(1),\qquad
 \sup_{a\in I}|Q_n'(a)|\leq C_I H_n(1).
\]
It follows that $\widetilde Q_n(\widehat a_n)-\widetilde Q_n(a_0)=\op(r_n^2)$,
and the same assertion holds for $Q_n$.
Write $d_t(a)=\{\widehat m_n(a)-1\}x_t+
                   \widehat\mu_n(a)-\mu$. The squared residual perturbation
in eq.~\eqref{eq:residual-perturbation} satisfies, uniformly on $I$,
\begin{align*}
 \mathcal E_n(a)=\sum_t w_a(x_t)^2d_t(a)^2
 &\leq C_I\left\{n\sup_{a\in I}|\widehat m_n(a)-1|^2
       +H_n(1)\sup_{a\in I}|\widehat\mu_n(a)-\mu|^2\right\}\\
 &=\Op(1).
\end{align*}
Cauchy--Schwarz therefore gives
\[
 \sup_{a\in I}|\widehat Q_n(a)-\widetilde Q_n(a)|
 \leq2\sqrt{\sup_{a\in I}\widetilde Q_n(a)\sup_{a\in I}\mathcal E_n(a)}
       +\sup_{a\in I}\mathcal E_n(a)
 =\Op(r_n)=\op(r_n^2).
\]
Together with eq.~\eqref{ad:squaredreward} and positivity of the limiting
clock, this proves
$\widehat Q_n(\widehat a_n)/Q_n(a_0)\to_p1$.
The joint limits in eqs.~\eqref{eq:strict-joint} and~\eqref{eq:boundary-joint},
equivalence with the fixed-offset estimator, and Slutsky's theorem now prove the two cases of
eq.~\eqref{ad:feasiblepivot}.

\textit{Step 3: The constrained estimator.}
For the constrained estimator the score ratio is exact, and
\[
 \sup_{a\in I}|\widehat\mu_{n,\mathrm c}(a)-\widehat\mu_n(a)|
 \leq\sup_{a\in I}|\widehat m_n(a)-1|\,
       \frac{n}{\inf_{a\in I}H_n(a)}
 =\Op(r_n^{-2}).
\]
Thus the same equivalence with the fixed-offset estimator holds.
For constrained residuals, replace $d_t(a)$ by
$d_{\mathrm c}(a)=\widehat\mu_{n,\mathrm c}(a)-\mu$. Since
$\sup_{a\in I}|d_{\mathrm c}(a)|=\Op(r_n^{-1})$, the squared residual
perturbation is again uniformly $\Op(1)$. The same residual argument
gives the corresponding constrained studentized limit in each regime.
\end{proof}

\subsection{Closed form of the variance regression}
\label{ad:optimal-proofs}

Set $w_t=w_1(X_{t-1})$ and $y_t=w_te_t^2$, where $e_t=e_t(1)$
as in Section~\ref{sec:adaptive}. The weighted
regression of Section~\ref{sec:adaptive} is equivalent to the ordinary
regression of $y_t$ on an intercept and $w_t$, whose intercept is
$\widehat{\sigma^2}_n$ and whose fitted value at $w=1$ is $\widehat b_n$.
Let $\widehat\gamma_n$ be the slope in this transformed regression,
which estimates $b-\sigma^2$. Writing $\bar y_n$ and $\bar w_n$ for the
sample means,
\begin{equation}
 \begin{aligned}
 \widehat\gamma_n
 &=\frac{\sum_{t=1}^n(w_t-\bar w_n)(y_t-\bar y_n)}
         {\sum_{t=1}^n(w_t-\bar w_n)^2},\\
 \widehat{\sigma^2}_n&=\bar y_n-\widehat\gamma_n\bar w_n,
 \qquad
 \widehat b_n=\bar y_n+(1-\bar w_n)\widehat\gamma_n.
 \end{aligned}
 \label{ad:variancepilot}
\end{equation}
Here $\bar w_n=H_n(1)/n$.

\subsection{Estimated offsets in the transient regime}
\label{ad:transient-proofs}

Under transience, $\widehat b_n$ converges in probability to
a random limit $B_\infty$, as shown in Step~1 of the proof of
Proposition~\ref{ad:transientoffset} below. The offset estimator
in eq.~\eqref{ad:offsetpilot} is discontinuous in $\widehat b_n$ at zero, so
applying the continuous mapping theorem requires
$\PP(B_\infty=0)=0$. The next lemma establishes the stronger property
that $B_\infty$ has no atoms.

\begin{lemma}[Absence of atoms in the transient variance limit]
\label{ad:variance-no-atoms}
Under Assumptions~\ref{ass:model} and~\ref{ad:fourthmoments}, suppose
$\tau>1$. The limits
\[
 H_{\infty,2}=\sum_{t\geq1}\frac{1}{(1+X_{t-1})^2},\qquad
 L_\infty^{\mathrm{var}}=\lim_{n\to\infty}\sum_{t=1}^n
 \frac{W_t^2-V(X_{t-1})}{(1+X_{t-1})^2}
\]
exist almost surely, with $0<H_{\infty,2}<\infty$ and
$|L_\infty^{\mathrm{var}}|<\infty$. Define
$B_\infty=b+L_\infty^{\mathrm{var}}/H_{\infty,2}$. Then
\begin{equation}
 \PP(B_\infty=b_0)=0\qquad\text{for every }b_0\in\R.
 \label{ad:transientnoatom}
\end{equation}
\end{lemma}

\begin{proof}
The existence and stated bounds for $H_{\infty,2}$ and
$L_\infty^{\mathrm{var}}$ follow from
\citet[Lemma~2.4]{WeiWinnicki1990} and
\citet[Lemma~2.11(c)]{Winnicki1991}, respectively.

\textit{Step 1: Holding probabilities along the transient path.}
The occupation results of
\citet[Corollary~2.11, Theorem~2.12, and Lemma~2.13]{WeiWinnicki1989}
give, almost surely,
\[
 X_t\to\infty,\qquad
 H_\infty:=\sum_{t\geq0}\frac{1}{1+X_t}=\infty,\qquad
 \sum_{t\geq0}(1+X_t)^{-3/2}<\infty.
\]
Let $P(x,y)$ be the transition kernel and $\vartheta_x=P(x,x)$.
The nondegenerate offspring law has mean one, so $\PP(Z=0)>0$.
Its maximal lattice span is therefore
$d=\gcd\{k\geq1:\PP(Z=k)>0\}$, and $Z\in d\mathbb Z$.
Write
$\mathcal R=\{s\in\{0,\ldots,d-1\}:\PP(\epsilon\equiv s\pmod d)>0\}$.
Every state observed after time zero has residue in $\mathcal R$.
For independent copies $Z_1,\ldots,Z_x$ of $Z$, the lattice local
central limit theorem \citep{SzewczakWeber2023} gives
\[
 \sup_k\PP\!\left(\sum_{i=1}^xZ_i=k\right)\leq Cx^{-1/2},
 \qquad
 \sqrt{x}\,\PP\!\left(\sum_{i=1}^xZ_i=x-j\right)
 \longrightarrow\frac{d}{\sigma\sqrt{2\pi}}
\]
for each fixed integer $j\geq0$ in the second limit along $x\equiv j\pmod d$.
Convolution with the immigration law gives the same upper bound for
$\vartheta_x$. For each $s\in\mathcal R$, choose one immigration value $j_s$
with positive probability and residue $s$. Retaining that term in the
convolution gives the lower bound. Since $\mathcal R$ is finite,
there are constants $c_\vartheta,C_\vartheta>0$ such that
\begin{equation}
 \frac{c_\vartheta}{\sqrt{1+x}}\leq \vartheta_x\leq\frac{C_\vartheta}{\sqrt{1+x}}
 \quad\text{for all sufficiently large }x\text{ with }x\bmod d\in\mathcal R.
 \label{ad:holdingprobabilities}
\end{equation}
The upper bound extends to every $x\geq0$ after enlarging $C_\vartheta$.
Moreover, $\bar\vartheta:=\sup_x \vartheta_x<1$, because the upper bound gives $\vartheta_x\to0$
and $\Var(X_1\mid X_0=x)=\sigma^2x+b>0$ excludes
 $\vartheta_x=1$ at
each finite state.

\textit{Step 2: Conditioning on the jump chain.}
Define the successive departure times and the embedded jump chain by
\[
 T_0=0,\qquad T_{j+1}=\inf\{t>T_j:X_t\ne X_{T_j}\},\qquad
 \Xi_j=X_{T_j},\qquad K_j=T_{j+1}-T_j-1.
\]
These times are finite almost surely. Conditional on
$\mathcal G=\sigma(\Xi_j:j\geq0)$, the holding counts $K_j$ are
independent and
\[
 \PP(K_j=k\mid\mathcal G)=(1-\vartheta_{\Xi_j})\vartheta_{\Xi_j}^{k},
 \qquad k=0,1,\ldots.
\]
Indeed, for $y\ne x$, the factorization
$\vartheta_x^kP(x,y)=[(1-\vartheta_x)\vartheta_x^k][P(x,y)/(1-\vartheta_x)]$ and the Markov
property separate each holding count from the next departure state.
The occupation results imply $\Xi_j\to\infty$ and
$\sum_j(1+\Xi_j)^{-3/2}<\infty$ almost surely. Also,
\[
 \E(H_\infty\mid\mathcal G)
 =\sum_{j\geq0}\frac{1}{(1-\vartheta_{\Xi_j})(1+\Xi_j)}
 \leq\frac{1}{1-\bar\vartheta}\sum_{j\geq0}\frac{1}{1+\Xi_j}.
\]
Since $H_\infty=\infty$ almost surely, the sum on the right must
diverge almost surely. The lower bound in
eq.~\eqref{ad:holdingprobabilities} therefore yields
$\sum_j\vartheta_{\Xi_j}=\infty$ almost surely; only the initial jump-chain
state can have residue outside $\mathcal R$.

Fix $b_0\in\R$ and put
\[
 \Phi_{b_0}(x,y)=\frac{(y-x-\mu)^2-\sigma^2x-b_0}{(1+x)^2},\qquad
 \phi_{b_0}(x)=\Phi_{b_0}(x,x)=\frac{\mu^2-\sigma^2x-b_0}{(1+x)^2}.
\]
The convergent path sum satisfies
\begin{equation}
 D_{b_0}:=\lim_{n\to\infty}\sum_{t=1}^n\Phi_{b_0}(X_{t-1},X_t)
 =L_\infty^{\mathrm{var}}+(b-b_0)H_{\infty,2}
 =(B_\infty-b_0)H_{\infty,2}.
 \label{ad:atomfunctional}
\end{equation}
Here $\phi_{b_0}(x)$ is nonzero for all sufficiently large $x$, since
$\sigma^2>0$, and $|\phi_{b_0}(x)|\leq C_{b_0}(1+x)^{-1}$. Consequently,
\[
 \E\!\left(\left.\sum_{j\geq0}|\phi_{b_0}(\Xi_j)|K_j\right|\mathcal G\right)
 =\sum_{j\geq0}|\phi_{b_0}(\Xi_j)|\frac{\vartheta_{\Xi_j}}{1-\vartheta_{\Xi_j}}
 \leq C_{b_0}'\sum_{j\geq0}(1+\Xi_j)^{-3/2}<\infty
\]
almost surely. Thus the holding-time series converges absolutely,
conditionally almost surely. Evaluating the path sum at completed
departures gives
\begin{equation}
 D_{b_0}=D_{b_0}^{\mathrm{jump}}(\Xi)+\sum_{j\geq0}\phi_{b_0}(\Xi_j)K_j,\qquad
 D_{b_0}^{\mathrm{jump}}(\Xi)=\sum_{j\geq0}\Phi_{b_0}(\Xi_j,\Xi_{j+1}).
 \label{ad:holdingdecomposition}
\end{equation}
The departure series converges because both $D_{b_0}$ and the holding-time
series converge. Its value is deterministic conditional on $\mathcal G$.

\textit{Step 3: Eliminating point masses.}
For independent real variables $\mathcal V_1,\ldots,\mathcal V_N$, the
Kolmogorov--Rogozin inequality
\citep[Theorem~1.1]{Juskevicius2024} states that
\[
 \mathcal Q_\eta\!\left(\sum_{j=1}^N\mathcal V_j\right)
 \leq C\left\{\sum_{j=1}^N[1-\mathcal Q_\eta(\mathcal V_j)]\right\}^{-1/2},
 \qquad
 \mathcal Q_\eta(\mathcal V)=\sup_u\PP(\mathcal V\in(u,u+\eta]),\quad \eta>0.
\]
Fix a jump chain satisfying the preceding almost-sure properties, and
let $I_N=\{0\leq j\leq N:\phi_{b_0}(\Xi_j)\ne0\}$.
Choose $\eta$ smaller than every $|\phi_{b_0}(\Xi_j)|$ with $j\in I_N$.
The conditional geometric law gives
$\mathcal Q_\eta(\phi_{b_0}(\Xi_j)K_j\mid\mathcal G)=1-\vartheta_{\Xi_j}$.
Hence the maximal point mass of the corresponding finite sum is at
most $C(\sum_{j\in I_N}\vartheta_{\Xi_j})^{-1/2}$. Adding the independent
remainder of the absolutely convergent series cannot increase this
maximal point mass. Since only finitely many indices are omitted and
$\sum_j\vartheta_{\Xi_j}=\infty$, letting $N\to\infty$ in
eq.~\eqref{ad:holdingdecomposition} gives
$\PP(D_{b_0}=0\mid\mathcal G)=0$ almost surely.
Eq.~\eqref{ad:atomfunctional} and $H_{\infty,2}>0$ now prove
eq.~\eqref{ad:transientnoatom}.
\end{proof}

\begin{proof}[Proof of Proposition~\ref{ad:transientoffset}]
\textit{Step 1: The variance regression and the offset limit.}
Retain $x_t$ and $w_t$ from above, abbreviate
$H_n=H_n(1)$ and $M_n=M_n(1)$, and set
\[
 H_{n,2}=\sum_{t=1}^nw_t^2,\qquad
 M_{n,2}=\sum_{t=1}^nw_t^2W_t,\qquad
 L_n^{\mathrm{var}}=\sum_{t=1}^nw_t^2\{W_t^2-V(x_t)\}.
\]
Lemma~\ref{ad:variance-no-atoms} gives the almost-sure limits of
$H_{n,2}$ and $L_n^{\mathrm{var}}$.
By \citet[Lemma~2.10(c)]{Winnicki1991}, $M_{n,2}$ also converges almost
surely to a finite limit.

For the initial fit at offset one, write
$d_m^0=\widehat m_n(1)-1$ and
$d_\mu^0=\widehat\mu_n(1)-\mu$.
\citet[Lemma~2.4 and Theorem~2.5]{WeiWinnicki1990} give
\[
 H_n=\Op(\log n),\qquad
 d_m^0=\Op(n^{-1}),\qquad
 d_\mu^0=\Op((\log n)^{-1/2}).
\]
The intercept normal equation then gives
$M_n=H_nd_\mu^0+d_m^0(n-H_n)=\Op(\sqrt{\log n})$.
Since $e_t=W_t-d_m^0x_t-d_\mu^0$ and $w_t^2x_t=w_t-w_t^2$,
\begin{align*}
 \Delta_n^{\mathrm{res}}&:=\sum_{t=1}^nw_t^2(e_t^2-W_t^2)\\
 &=-2d_m^0(M_n-M_{n,2})-2d_\mu^0M_{n,2}
       +\sum_{t=1}^nw_t^2(d_m^0x_t+d_\mu^0)^2=\op(1).
\end{align*}
Here the last sum is bounded by
$2n(d_m^0)^2+2H_{n,2}(d_\mu^0)^2=\op(1)$.
The variance-regression normal equation gives
\begin{equation}
 (\widehat b_n-b)H_{n,2}
 =L_n^{\mathrm{var}}+\Delta_n^{\mathrm{res}}-(\widehat{\sigma^2}_n-\sigma^2)
                   (H_n-H_{n,2}).
 \label{ad:transientvariancenormal}
\end{equation}
The regression design is nonsingular with probability tending to one,
since its determinant divided by $n$ is
$H_{n,2}-H_n^2/n\to_p H_{\infty,2}>0$.
By \citet[Theorem~3.10]{Winnicki1991},
$\widehat{\sigma^2}_n-\sigma^2=\Op(n^{-1/2})$, so the last term
in eq.~\eqref{ad:transientvariancenormal} is $\op(1)$.
Dividing by $H_{n,2}$ proves
\begin{equation}
 \widehat b_n\to_p B_\infty.
 \label{ad:transientvariancelimit}
\end{equation}

The map in eq.~\eqref{ad:offsetpilot} is continuous at
$(B_\infty,\sigma^2)$ whenever $B_\infty\ne0$.
Lemma~\ref{ad:variance-no-atoms} gives $\PP(B_\infty=0)=0$,
so the continuous mapping theorem yields
\begin{equation}
 \widehat a_n\to_p A_\infty:=
 \begin{cases}
 B_\infty/\sigma^2,&B_\infty>0,\\
 1,&B_\infty\leq0,
 \end{cases}
 \qquad 0<A_\infty<\infty\quad\text{almost surely}.
 \label{ad:transientoffsetlimit}
\end{equation}
In particular,
\begin{equation}
 \widehat a_n+\widehat a_n^{-1}=\Op(1).
 \label{ad:transientcompact}
\end{equation}

\textit{Step 2: Uniform comparison of scores and clocks.}
Put $c=\mu-\sigma^2/2>0$ and use $\widetilde Q_n(a)$ from
eq.~\eqref{eq:squared-score}.
The same results of Wei and Winnicki, together with the transient
residual argument in Appendix~\ref{app:residual-studentization}, give
\begin{equation}
 \frac{H_n(1)}{\log n}\to_p c^{-1},\qquad
 M_n(1)=\Op(\sqrt{\log n}),\qquad
 \frac{\widetilde Q_n(1)}{\log n}\to_p\frac{\sigma^2}{c}.
 \label{ad:transientbaseline}
\end{equation}
Eq.~\eqref{eq:transient-offset-bounds}, proved under
Assumption~\ref{ass:model}, bounds the score, clock, and squared-score
differences uniformly over every fixed compact interval
$I\subset(0,\infty)$.
By eq.~\eqref{ad:transientcompact}, a compact $I$ can be chosen to contain
$\widehat a_n$ with arbitrarily high limiting probability. Thus
\begin{align}
 H_n(\widehat a_n)-H_n(1)&=\Op(1),&
 M_n(\widehat a_n)-M_n(1)&=\Op(1),\label{ad:transientscorecompare}\\
 \widetilde Q_n(\widehat a_n)-\widetilde Q_n(1)&=\Op(1),&
 \sum_tw_{\widehat a_n}(x_t)^2&=\Op(1).
 \label{ad:transientquadraticcompare}
\end{align}

\textit{Step 3: Equivalence of the drift estimators.}
Substituting eq.~\eqref{eq:critical-functional-inputs} into
eq.~\eqref{eq:slope-normal}, together with
eqs.~\eqref{ad:transientbaseline}--\eqref{ad:transientscorecompare},
shows that the design at $\widehat a_n$ is nonsingular with
probability tending to one and that
$\widehat m_n(\widehat a_n)-1=\Op(n^{-1})$.
The two score ratios differ by
\[
 \frac{M_n(\widehat a_n)}{H_n(\widehat a_n)}
 -\frac{M_n(1)}{H_n(1)}=\Op((\log n)^{-1}).
\]
Both slope contributions in eq.~\eqref{eq:reduction} are
$\Op((\log n)^{-1})$. This gives the first limit below; combining
it with eq.~\eqref{eq:transient} gives the second:
\begin{align}
 \sqrt{\log n}\{\widehat\mu_n(\widehat a_n)
                   -\widehat\mu_n(1)\}&\to_p0,
 \label{ad:transientequivalence}\\
 \sqrt{\log n}\{\widehat\mu_n(\widehat a_n)-\mu\}
 &\Rightarrow
 N\!\left(0,\sigma^2\left(\mu-\frac{\sigma^2}{2}\right)\right).
 \label{ad:transientnormal}
\end{align}

\textit{Step 4: Residual studentization.}
Put $d_m=\widehat m_n(\widehat a_n)-1=\Op(n^{-1})$ and
$d_\mu=\widehat\mu_n(\widehat a_n)-\mu=\Op((\log n)^{-1/2})$.
By eq.~\eqref{ad:transientquadraticcompare},
\[
 \mathcal E_n(\widehat a_n)=\sum_tw_{\widehat a_n}(x_t)^2(d_mx_t+d_\mu)^2
 \leq2nd_m^2+2d_\mu^2\sum_tw_{\widehat a_n}(x_t)^2
 =\Op((\log n)^{-1}).
\]
Cauchy--Schwarz gives
\[
 |\widehat Q_n(\widehat a_n)-\widetilde Q_n(\widehat a_n)|
 \leq2\sqrt{\widetilde Q_n(\widehat a_n)\mathcal E_n(\widehat a_n)}
       +\mathcal E_n(\widehat a_n)=\Op(1).
\]
Thus $\widehat Q_n(\widehat a_n)/\log n\to_p\sigma^2/c$.
Together with eqs.~\eqref{ad:transientbaseline}--\eqref{ad:transientscorecompare}
and~\eqref{ad:transientnormal}, this proves eq.~\eqref{ad:transientstudentized}.
The constrained estimator differs from the joint estimator by
$\Op((\log n)^{-1})$; setting $d_m=0$ above proves its studentized limit.
\end{proof}

\end{document}